\documentclass[11pt,a4paper,reqno]{amsart}
\usepackage{multirow}
\usepackage{cite}
\usepackage{amsmath}
\usepackage{amsfonts}
\usepackage{amssymb}
\usepackage{amsthm}
\usepackage{newlfont}
\usepackage{a4}
\usepackage{enumerate}
\newcommand{\doi}[1]{doi:#1}

\usepackage{graphicx,color}

\theoremstyle{definition}
\newtheorem{defn}{Definition}[section]
\newtheorem{thm}[defn]{Theorem}

\newtheorem{tvr}[defn]{Proposition}

\theoremstyle{remark}
\newtheorem{example}{Example}[section]
\newcommand{\id}{\mathfrak{1}}
\usepackage{bbm}

\newlength{\defbaselineskip}
\newcommand{\setlinespacing}[1]%
           {\setlength{\baselineskip}{#1 \defbaselineskip}}

\newcommand{\hart}{\zeta}
\renewcommand{\i}{\mathrm{i}}
\newcommand{\map}{\rightarrow}

\newcommand{\q}{\quad}

\renewcommand{\epsilon}{\varepsilon}

\newcommand{\ep}{\varepsilon}
\newcommand{\la}{\lambda}
\newcommand{\al}{\alpha}
\newcommand{\om}{\omega}
\renewcommand{\rho}{\varrho}
\renewcommand{\phi}{\varphi}

\newcommand{\R}{{\mathbb{R}}}

\newcommand{\N}{{\mathbb N}}

\newcommand{\Z}{\mathbb{Z}}
\newcommand{\C}{\mathbb{C}}

\newcommand{\set}[2]{\left\{#1 \, |\, #2 \right\}}

\newcommand{\abs}[1]{\left\vert#1\right\vert}
\newcommand{\wt}{\widetilde}

\newcommand{\sca}[2]{\langle #1,\, #2\rangle}

\begin{document}

\title[Cosine transforms generalized to honeycomb lattice]
{Discrete cosine and sine transforms generalized to honeycomb lattice}

\author[J. Hrivn\'{a}k]{Ji\v{r}\'{i} Hrivn\'{a}k$^{1}$}
\author[L. Motlochov\'{a}]{Lenka Motlochov\'{a}$^{1}$}

\date{\today}
\begin{abstract}\small The discrete cosine and sine transforms are generalized to a triangular fragment of the honeycomb lattice. The honeycomb point sets are constructed by subtracting the root lattice from the weight lattice points of the crystallographic root system $A_2$. 
The two-variable orbit functions of the Weyl group of $A_2$, discretized simultaneously on the weight and root lattices, induce a novel parametric family of extended Weyl orbit functions. The periodicity and von Neumann and Dirichlet boundary properties of the extended Weyl orbit functions are detailed. Three types of discrete complex Fourier-Weyl transforms and real-valued Hartley-Weyl transforms are described. Unitary transform matrices and interpolating behaviour of the discrete transforms are exemplified. Consequences of the developed discrete transforms for transversal eigenvibrations of the mechanical graphene model are discussed. 
\end{abstract}

\maketitle
\noindent
$^1$ Department of Physics, Faculty of Nuclear Sciences and Physical Engineering, Czech Technical University in Prague, B\v{r}ehov\'a~7, CZ--115 19 Prague, Czech Republic
\vspace{10pt}

\noindent
\textit{E-mail:} jiri.hrivnak@fjfi.cvut.cz, lenka.motlochova@fjfi.cvut.cz

\bigskip

\bigskip


\section{Introduction}
The purpose of this article is to generalize the discrete cosine and sine transforms \cite{Brit,strang} to a finite fragment of the honeycomb lattice. Two-variable (anti)symmetric complex-valued Weyl and the real-valued Hartley-Weyl orbit functions \cite{KP1,KP2,HJedis,HMdis} are modified by six extension coefficients sequences and the resulting family of the extended functions contains the set of the discretely orthogonal honeycomb orbit functions.  
The triangularly shaped fragment of the honeycomb lattice forms the set of nodes of the honeycomb Fourier-Weyl and Hartley-Weyl discrete transforms. 

The one-variable discrete cosine and sine transforms and their multivariate concatenations constitute the backbone of the digital data processing methods \cite{Brit, strang}. The periodicity and (anti)symmetry of the cosine and sine functions represent intrinsic symmetry properties inherited from the (anti)symmetrized one-variable exponential functions. These symmetry characteristics, essential for data processing applications, restrict the trigonometric functions to a bounded interval and induce boundary behaviour of these functions at the interval endpoints. The boundary behaviour of the most ubiquitous two-dimensional cosine and sine transforms is directly induced by the boundary features of their one-dimensional versions \cite{Brit}.  Two-variable symmetric and antisymmetric orbit functions of the crystallographic reflection group $A_2$, confined by their symmetries to their fundamental domain of an equilateral triangle shape, satisfy similar Dirichlet and von Neumann boundary conditions \cite{KP1,KP2}. These fundamental boundary properties are preserved by the novel parametric sets of the extended Weyl and Hartley-Weyl orbit functions. Narrowing the classes of the extended Weyl and Hartley-Weyl orbit functions by three non-linear conditions, continua of parametric sets of the discretely orthogonal honeycomb orbit functions and the corresponding discrete transforms are found.    

Discrete transforms of the Weyl orbit functions over finite sets of points and their applications are, of recent, intensively studied \cite{HP,HMPdis,HMPcub,HW1,HW2,xuAd}. The majority of these methods arise in connection with point sets taken as the fragments of the refined dual-weight lattice. Discrete orthogonality and transforms of the Weyl orbit functions over a fragment of the refined weight lattice are formulated in connection with the conformal field theory in \cite{HW2}. Discrete orthogonality and transforms of the Hartley-Weyl orbit functions over a fragment of the refined dual weight lattice are obtained in a more general theoretical setting in \cite{HJedis} and explicit formulations of discrete orthogonality and transforms of both Weyl and Hartley-Weyl functions over a fragment of the refined dual root lattice are achieved in \cite{HMdis}. For the crystallographic root system $A_2$, the dual weight and weight lattices as well as the dual root and root lattices coincide. Since the honeycomb lattice is not in fact a lattice in a strict mathematical sense, its points in the context of the root system $A_2$ are constructed by subtracting the root lattice from the weight lattice. Recent explicit formulations of the discrete orthogonality relations of the Weyl and Hartley-Weyl orbit functions over both root and weight lattices of $A_2$ permit construction of the discrete honeycomb functions and transforms. The introduced extended Weyl and Hartley-Weyl orbit functions, obeying the three non-linear conditions, represent a unique novel discretely orthogonal parametric systems of functions over the subtractively constructed honeycomb point sets. The counting formulas for the numbers of elements in the point and label sets in \cite{HJedis} and \cite{HMdis} guarantee existence of both Weyl and Hartley-Weyl discrete Fourier transforms over the honeycomb point sets.
  
The potential of the developed Fourier-like discrete honeycomb transforms lies in the data processing methods on the triangular fragment of the honeycomb lattice as well as in theoretical description of the properties of the graphene material \cite{Castro,Cooper}. Both transversal and longitudinal vibration modes of the graphene are regularly studied assuming Born--von K\'{a}rm\'{a}n periodic boundary condidions \cite{Castro,CsTi,DrSa,Falkovsky,Sahoo}. Analysis of wave functions of the electron in a triangular graphene quantum dot \cite{Roz} leads to discrete functions which obey Dirichlet conditions on the armchair-type boundary. Special cases of the presented honeycomb orbit functions represent analogous vibration modes that satisfy either Dirichlet or von Neumann conditions on the boundary of the same type. Since spectral analysis provided by the developed discrete honeycomb transforms enjoys similar boundary properties as the 2D discrete cosine and sine transforms, the output analysis of the graphene-based sensors \cite{Shiva,Jamlos} by the honeycomb transforms offers similar data processing potential. The honeycomb discrete transforms provide novel possibilities for study of distribution of spectral coefficients \cite{Lam}, image watermarking \cite{Hernandez}, encryption \cite{Liu}, and compression \cite{Fracastoro} techniques.   

The paper is organized as follows. In Section 2, a self-standing review of $A_2$ inherent lattices and extensions of  the corresponding crystallographic reflection group is included. In Section 3, the fundamental finite fragments of the honeycomb lattice and the weight lattice are described. In Section 4, two types of discrete orthogonality relations of the Weyl and Hartley-Weyl orbit functions are recalled. In Section 5, definitions of the extended and honeycomb orbit functions together with the proof of their discrete orthogonality are detailed. In Section 6, three types of the honeycomb orbit functions are exemplified and depicted. In Section 7, the four types of the discrete honeycomb lattice transforms and unitary matrices of the normalized transforms are formulated. Comments and follow-up questions are covered in the last section.

\section{Infinite extensions of Weyl group}

\subsection{Root and weight lattices}\

Notation, terminology, and pertinent facts about the simple Lie algebra $A_2$ stem from the theory of Lie algebras and crystallographic root systems \cite{H2,Bour}. The $\al-$basis of the $2-$dimensional Euclidean space $\R^2$, with its scalar product denoted by $\sca{\,}{\,}$, comprises vectors $\alpha_1,\alpha_2$ characterized by their lenghts and relative angle as 
\begin{equation*}
\sca{\al_1}{\al_1}=\sca{\al_2}{\al_2}=2,\q\sca{\al_1}{\al_2}=-1.
\end{equation*}
In the context of simple Lie algebras and root systems, the vectors $\alpha_1$ and $\alpha_2$ form the simple roots of~$A_2$. 
In addition to the $\alpha-$basis of the simple roots, it is convenient to introduce $\omega-$basis of $\R^2$ of vectors $\omega_1$ and $\omega_2$, called the fundamental weights, satisfying
\begin{equation}\label{Zdual}
\left\langle \om_i,\alpha_j\right\rangle=\delta_{ij}\,,\quad i,j\in\{1,2\}. 	
\end{equation}

The vectors $\om_1$ and $\om_2$, written in the $\alpha-$basis, are given as
\begin{equation*}
\omega_1=\tfrac23\al_1+\tfrac13\al_2,\q \omega_2=\tfrac13\al_1+\tfrac23\al_2,
\end{equation*}
and the inverse transform is of the form
\begin{equation*}
\al_1=2\om_1-\om_2,\q\al_2=-\om_1+2\om_2.
\end{equation*}
The lenghts and relative angle of the vectors $\om_1$ and $\om_2$ are determined by
\begin{equation}\label{scaom}
\sca{\om_1}{\om_1}=\sca{\om_2}{\om_2}=\tfrac23,\q\sca{\om_1}{\om_2}=\tfrac13,	
\end{equation}
and the scalar product of two vectors in $\om-$basis $x=x_1\omega_1+x_2\omega_2$ and $y=y_1\omega_1+y_2\omega_2$ is derived as  
\begin{equation}\label{scprod}
\sca{x}{y}=\tfrac13 (2 x_1y_1 + x_1y_2 + x_2y_1 + 2 x_2y_2).
\end{equation}

The lattice $Q\subset\R^2$, referred to as the root lattice, comprises all integer linear combinations of the $\al-$basis,
$$Q = \Z \alpha_1 + \Z \alpha_2.$$
The lattice $P\subset\R^2$, called the weight lattice, consists of all integer linear combinations of the $\om-$basis,
$$P = \Z \om_1 + \Z \om_2.$$
The root lattice $P$ is disjointly decomposed into three shifted copies of the root lattice $Q$ as
\begin{equation}\label{Pdec}
P=Q\cup \{\om_1+Q\} \cup \{\om_2+Q\}.	
\end{equation}

The reflections $r_i$, $i=1,2$, which fix the hyperplanes orthogonal to $\alpha_i$ and pass through the origin are linear maps expressed for any $x\in\R^2$ as
\begin{equation}\label{ri}
	r_ix=x-\langle \alpha_i,x\rangle\alpha_i.
\end{equation}
The associated Weyl group $W$ of $A_2$ is a finite group generated by the reflections $r_1$ and $r_2$. The Weyl group orbit $Wx$, constituted by $W-$images of the point $x\equiv(x_1,x_2)=x_1\om_1+x_2\om_2$, is given in $\om-$basis as
\begin{equation}\label{orbit}
Wx=\{(x_1,x_2),(-x_1, x_1 + x_2), (-x_1 - x_2, x_1), (-x_2, -x_1), (x_2, -x_1 - x_2), (x_1 + x_2, -x_2)\}.
\end{equation}
The root lattice $Q$ and the weight lattice $P$ are Weyl group invariant, 
$$ WP=P,\q WQ=Q.$$ 

\subsection{Affine Weyl group}\

The affine Weyl group of $A_2$ extends the Weyl group $W$ by shifts by vectors from the root lattice $Q$,
$$W_Q^{\mathrm{aff}}=  Q \rtimes W.$$ 
Any element $T(q)w\in W_Q^{\mathrm{aff}}$ acts on any $x\in\R^2$ as $$T(q)w\cdot x= wx+q.$$
The fundamental domain $F_Q$ of the action of $W_Q^{\mathrm{aff}}$ on $\R^2$, which consists of exactly one point from each $W_Q^{\mathrm{aff}}-$orbit, is a triangle with vertices $\left\{ 0, \om_1,\om_2 \right\}$,
\begin{align}\label{F}
F_Q&=\left\lbrace x_1\om_1+x_2\om_2\, |\, x_1,x_2\ge 0, x_1+x_2\le 1\right\rbrace .
\end{align}

For any $M\in\N$, the point sets $F_{P,M}$ and $F_{Q,M}$ are defined as finite fragments of the refined lattices $\tfrac{1}{M} P$ and $\tfrac{1}{M} Q$ contained in $F_Q$,
\begin{align}
	F_{P,M}&=\tfrac{1}{M} P\cap F_Q, \label{FPM}\\
	F_{Q,M}&=\tfrac{1}{M} Q\cap F_Q.\label{FQM}
\end{align}
The point sets $F_{P,M}$ and $F_{Q,M}$ are of the following explicit form,
\begin{equation}\label{gridF}
\begin{aligned}
F_{P,M}&=\set{\tfrac{s_1}{M}\omega_1+\tfrac{s_2}{M}\omega_2}{s_0,s_1,s_2\in\Z^{\ge0},s_0+s_1+s_2=M},\\
F_{Q,M}&=\set{\tfrac{s_1}{M}\omega_1+\tfrac{s_2}{M}\omega_2}{s_0,s_1,s_2\in\Z^{\ge0},s_0+s_1+s_2=M,s_1+2s_2=0\;\mathrm{mod}\,3}.
\end{aligned}
\end{equation}
Note that the points from $F_{P,M}$ are described by the coordinates from \eqref{gridF} as
\begin{equation}\label{kac}
s=[s_0,s_1,s_2]\in F_{P,M}	
\end{equation}
and the set $F_{Q,M}\subset F_{P,M}$ contains only such points from $F_{P,M}$, which satisfy the additional condition $s_1+2s_2=0\;\mathrm{mod}\,3$. The number of points in the point sets $F_{P,M}$ and $F_{Q,M}$ are calculated in \cite{HP,HMdis} as
\begin{align}\label{numb1}
\abs{F_{P,M}}&=\tfrac12(M^2+3M+2),\\
\abs{F_{Q,M}}&=\begin{cases}
\tfrac16(M^2+3M+6)&M=0\;\mathrm{mod}\, 3,\\
\tfrac16(M^2+3M+2)&\text{otherwise}.\label{numb2}
\end{cases}
\end{align}

Interiors of the point sets $F_{P,M}$ and $F_{Q,M}$ contain the grid points from the interior of $F_Q$,
\begin{align}
	\wt{F}_{P,M} &=\tfrac{1}{M} P\cap\mathrm{int}(F_Q),  \label{FPMint}\\
\wt{F}_{Q,M} &=\tfrac{1}{M} Q\cap\mathrm{int}(F_Q).\label{FQMint}
\end{align}
The explicit forms of the point sets interiors $\wt F_{P,M}$ and $\wt F_{Q,M}$ are the following,
\begin{equation}\label{gridintF}
\begin{aligned}
\wt{F}_{P,M}&=\set{\tfrac{s_1}{M}\omega_1+\tfrac{s_2}{M}\omega_2}{s_0,s_1,s_2\in\N,s_0+s_1+s_2=M},\\
\wt{F}_{Q,M}&=\set{\tfrac{s_1}{M}\omega_1+\tfrac{s_2}{M}\omega_2}{s_0,s_1,s_2\in\N,s_0+s_1+s_2=M,s_1+2s_2=0\;\mathrm{mod}\,3},
\end{aligned}
\end{equation}
and the counting formulas from \cite{HP,HMdis} calculate the number of points for $M>3$ as
\begin{align}
\abs{\wt{F}_{P,M}}&=\tfrac12(M^2-3M+2), \label{numb3}\\
\abs{\wt{F}_{Q,M}}&=\begin{cases}
\tfrac16(M^2-3M+6)&M=0\;\mathrm{mod}\, 3,\\
\tfrac16(M^2-3M+2)&\text{otherwise}.
\end{cases}\label{numb4}
\end{align}
The point sets $F_{P,M}$ and $F_{Q,M}$ are for $M=7$ depicted in Figure \ref{obrbody}.
\begin{figure}
\includegraphics{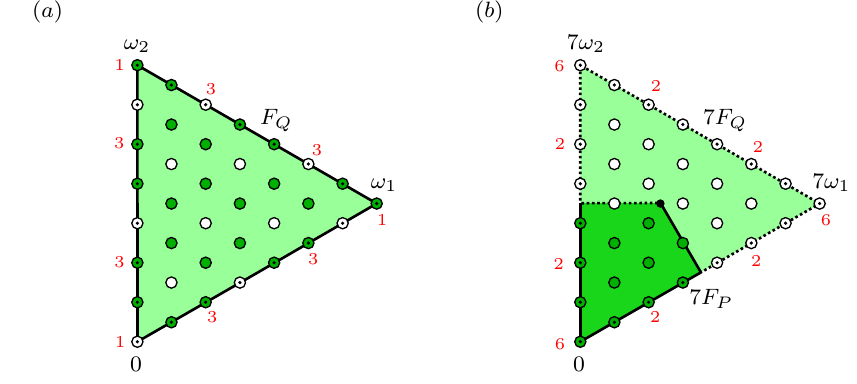}
\caption{\small $(a)$ The fundamental domain $F_Q$ is depicted as the green triangle containing 36 dark green and white nodes that represent the points of the set $F_{P,7}$. The elements of $F_{Q,7}$ are displayed as 12 white nodes. Omitting the dotted boundary points from $F_{P,7}$ and $F_{Q,7}$ yields 15 points of the interior set $\wt{F}_{P,7}$ and the 5 points of the set $\wt{F}_{Q,7}$, respectively. The numbers $1,3$ assigned to the nodes illustrate the values of the discrete function $\ep(s)$. $(b)$ The dark green kite-shaped domain $7F_P$ is contained in the lighter green triangle, which depicts the domain $7F_Q$. The white and dark green nodes represent 36 weights of the weight set $7\Lambda_{Q,7}$, the dark green nodes represent 12 weights of the set $\Lambda_{P,7}$. Omitting the dotted boundary points from $\Lambda_{Q,7}$ and $\Lambda_{P,7}$ yields 15 points of the interior set $\wt{\Lambda}_{Q,7}$ and the 5 points of the set $\wt{\Lambda}_{P,7}$, respectively. The numbers $2,6$ assigned to the nodes illustrate the values of the discrete function $h_7(\lambda)$.}
\label{obrbody}
\end{figure}
A discrete function $\ep:F_{P,M}\map \N$ is defined by its values on coordinates \eqref{kac}  of $s\in F_{P,M}$ in Table~\ref{eps}.
\bgroup
\def\arraystretch{1.3}
\begin{table}
\begin{tabular}{|c||c|c|c|c|c|c|c|}
\hline
$ $&$[s_0,s_1,s_2]$&$[0,s_1,s_2]$&$[s_0,0,s_2]$&$[s_0,s_1,0]$&$[0,0,s_2]$&$[0,s_1,0]$&$[s_0,0,0]$\\ \hline\hline
$\ep\left(s\right)$&6&3&3&3&1&1&1\\\hline
\end{tabular}
\bigskip
\caption{{\small The values of the function $\ep$ on coordinates \eqref{kac} of $s\in F_{P,M}$ with $s_0,s_1,s_2\neq 0.$}}
\label{eps}
\end{table}
\egroup

\subsection{Extended affine Weyl group}\

The extended affine Weyl group of $A_2$ extends the Weyl group $W$ by shifts by vectors from the weight lattice $P$,
$$W_P^{\mathrm{aff}}=  P \rtimes W.$$ 
Any element $T(p)w\in W_P^{\mathrm{aff}}$ acts on any $x\in\R^2$ as $$T(p)w\cdot x= wx+p.$$
For any $M\in\N$, the abelian group $\Gamma_M \subset W_P^{\mathrm{aff}}$, $$\Gamma_M=\{\gamma_0,\gamma_1,\gamma_2\},$$ is a finite cyclic subgroup of $W_P^{\mathrm{aff}}$ with its three elements given explicitly by 
\begin{equation}\label{gammaexp}
\gamma_0=T(0)1,\q\gamma_1=T(M\om_1)r_1r_2,\q\gamma_2=T(M\om_2)(r_1r_2)^2.	
\end{equation}

The fundamental domain $F_P$ of the action of $W_P^{\mathrm{aff}}$ on $\R^2$, which consists of exactly one point from each $W_P^{\mathrm{aff}}-$orbit, is a subset of $F_Q$ in the form of a kite given by
\begin{align*}
F_P=&\left\lbrace x_1\om_1+x_2\om_2\in F_Q\, |\,(2x_1+x_2<1,x_1+2x_2<1)\,\vee\,(2x_1+x_2=1,x_1\ge x_2)\right\rbrace,
\end{align*}
For any $M\in\N$, the weight sets $\Lambda_{Q,M}$ and $\Lambda_{P,M}$ are defined as finite fragments of the lattice $P$ contained in the magnified fundamental domains $MF_Q$ and $MF_P$, respectively,
\begin{align}
	\Lambda_{Q,M}=  P\cap MF_Q,\label{LQM}\\
	\Lambda_{P,M}= P\cap MF_P. \label{LPM}
\end{align}

The weight set $\Lambda_{Q,M}$ is of the following explicit form,
\begin{align*}
\Lambda_{Q,M}=&\set{\la_1\omega_1+\la_2\omega_2}{\la_0,\la_1,\la_2\in\Z^{\ge0},\la_0+\la_1+\la_2=M}
\end{align*}
and thus, the points from $\Lambda_{Q,M}$ are described as
\begin{equation}\label{kacla}
\la=[\la_0,\la_1,\la_2]\in \Lambda_{Q,M}.	
\end{equation}
The weight set $\Lambda_{P,M}$ is of the explicit form,
\begin{align}
\Lambda_{P,M}=&\set{[\la_0,\la_1,\la_2]\in \Lambda_{Q,M}}{(\la_0>\la_1,\la_0>\la_2)\,\vee\,(\la_0=\la_1\geq \la_2)}\label{gridLa}.
\end{align}
The numbers of points in the weight sets $\Lambda_{Q,M}$ and $\Lambda_{P,M}$ are proven in \cite{HP,HMdis} to coincide with the number of points in $F_{P,M}$ and $F_{Q,M}$, respectively,
\begin{equation}
\abs{\Lambda_{Q,M}}=\abs{F_{P,M}},\q\abs{\Lambda_{P,M}}=\abs{F_{Q,M}}.\label{numbla}
\end{equation}
The action of the group $\Gamma_M$ on a weight $[\la_0,\la_1,\la_2]\in \Lambda_{Q,M}$ coincides with a cyclic permutation of the coordinates $[\la_0,\la_1,\la_2]$,
\begin{equation}\label{cycl}
\gamma_0[\la_0,\la_1,\la_2]=[\la_0,\la_1,\la_2],\q \gamma_1[\la_0,\la_1,\la_2]=[\la_2,\la_0,\la_1],\q\gamma_2[\la_0,\la_1,\la_2]=[\la_1,\la_2,\la_0],	
\end{equation}
and the weight set $\Lambda_{Q,M}$ is tiled by the images of $\Lambda_{P,M}$ under the action of $\Gamma_M$, 
\begin{equation}\label{gam}\Lambda_{Q,M}=\Gamma_M\Lambda_{P,M}.\end{equation}
 
The subset $\Lambda_M^\mathrm{fix}\subset\Lambda_{P,M}$ contains only the points stabilized by the entire $\Gamma_M$, $$\Lambda_M^{\mathrm{fix}}=\set{\la\in\Lambda_{P,M}}{\Gamma_M\la=\la}.$$ Note that there exists at most one point $\la=[\la_0,\la_1,\la_2]$ from $\Lambda_{P,M}$, which is fixed by $\Gamma_M$. From relation \eqref{cycl}, such a point satisfies $\la_0=\la_1=\la_2=M/3$ and, consequently, the set $\Lambda^\mathrm{fix}_M$ is empty if $M$ is not divisible by $3$, otherwise it has exactly one point,
\begin{equation}\label{numb5}
\abs{\Lambda^\mathrm{fix}_M}=\begin{cases}
1&M=0 \;\mathrm{mod}\,3,\\
0&\mathrm{otherwise}.
\end{cases}
\end{equation}

Interiors $\wt{\Lambda}_{Q,M}$ and $\wt{\Lambda}_{P,M}$ of the weight sets $\Lambda_{Q,M}$ and $\Lambda_{P,M}$ contain only points belonging to the interior of the magnified fundamental domain $MF_Q$,
\begin{align}
\wt{\Lambda}_{Q,M} &= P\cap\mathrm{int}(MF_Q),\label{LQMint}\\
\wt{\Lambda}_{P,M}&= P\cap MF_P\cap\mathrm{int}(MF_Q). \label{LPMint}
\end{align}
The explicit forms of the interiors of the weight sets are given as
\begin{align}
\wt{\Lambda}_{Q,M}=&\set{[\la_0,\la_1,\la_2]\in \Lambda_{Q,M}}{\la_0,\la_1,\la_2\in\N},\nonumber\\
\wt\Lambda_{P,M}=&\set{[\la_0,\la_1,\la_2]\in\wt\Lambda_{Q,M}}{(\la_0>\la_1,\la_0>\la_2)\,\vee\,(\la_0=\la_1\geq \la_2)}.\label{gridintLa}
\end{align}
The numbers of weights in the interior weight sets $\wt\Lambda_{Q,M}$ and $\wt\Lambda_{P,M}$ are proven in \cite{HP,HMdis} to coincide with the number of points in interiors $\wt F_{P,M}$ and $\wt F_{Q,M}$, respectively,
\begin{equation}\label{numbintla}
\abs{\wt\Lambda_{Q,M}}=\abs{\wt{F}_{P,M}},\q\abs{\wt\Lambda_{P,M}}=\abs{\wt{F}_{Q,M}}.\end{equation}
The weight sets $\Lambda_{P,7}$ and $\Lambda_{Q,7}$ are depicted in Figure \ref{obrbody}.

A discrete function $h_M:\Lambda_{Q,M}\map \N $ is defined by its values on coordinates \eqref{kacla}  of $\la\in \Lambda_{Q,M}$ in Table~\ref{hM}. The function $h_M$ depends only on the number of zero-valued coordinates and thus, is invariant under permutations of $[\la_0,\la_1,\la_2]$,
\begin{equation}\label{invgam}
h_M(\gamma\la)=h_M(\la),\q\gamma\in\Gamma_M.
\end{equation}

\bgroup
\def\arraystretch{1.3}
\begin{table}
\begin{tabular}{|c||c|c|c|c|c|c|c|}
\hline
$ $&$[\la_0,\la_1,\la_2]$&$[0,\la_1,\la_2]$&$[\la_0,0,\la_2]$&$[\la_0,\la_1,0]$&$[0,0,\la_2]$&$[0,\la_1,0]$&$[\la_0,0,0]$\\ \hline\hline
$h_M\left(\la\right)$&1&2&2&2&6&6&6\\\hline
\end{tabular}
\bigskip
\caption{{\small 
The values of the function $h_M$ on coordinates \eqref{kacla} of $\la\in\Lambda_{Q,M}$ with $\la_0,\la_1,\la_2\neq 0.$}}
\label{hM}
\end{table}
\egroup

\section{Point and weight sets}

\subsection{Point sets $H_M$ and $\wt H_M$}\

For any $M\in\N$, the point set $H_M$ is defined as a finite fragment of the honeycomb lattice $\tfrac{1}{M}(P\setminus Q)$ contained in $F_Q$,
\begin{equation*}
H_M=\tfrac{1}{M}(P\setminus Q)\cap F_Q.
\end{equation*}
Equivalently, the honeycomb lattice fragment $H_M$ is obtained from the point set $F_{P,M}$ by omitting the points of $F_{Q,M}$,
\begin{equation}\label{HM}
H_M=F_{P,M}\setminus F_{Q,M}.
\end{equation}
Introducing the following two point sets,
\begin{align}
H_M^{(1)}=&\tfrac{1}{M}(\om_1+Q)\cap F_Q, \label{HM1}\\
H_M^{(2)}=&\tfrac{1}{M}(\om_2+Q)\cap F_Q, \label{HM2} 
\end{align}
their disjoint union coincides due to \eqref{Pdec} with the point set $H_M$,
\begin{equation}\label{HMdis}
H_M=H_M^{(1)}\cup H_M^{(2)}.
\end{equation}
The explicit description of $H_M$ is directly derived from \eqref{gridF} and \eqref{HM},
$$H_M=\set{\tfrac{s_1}{M}\om_1+\tfrac{s_2}{M}\om_2}{s_0,s_1,s_2\in\Z^{\geq0},s_0+s_1+s_2=M,s_1+2s_2\neq 0\,\mathrm{mod}\,3}.$$
\begin{tvr}
The number of points in the point set $H_M$ is given by
\begin{equation}\label{numH}
\abs{H_M}=\begin{cases}
\tfrac13(M^2+3M)&M=0\;\mathrm{mod}\,3,\\
\tfrac{1}{3}(M^2+3M+2)& \text{otherwise}.
\end{cases}
\end{equation}
\end{tvr} 
\begin{proof}
Relation \eqref{HM} implies that the number of points in $H_M$ is calculated as
$$\abs{H_M}=\abs{F_{P,M}}-\abs{F_{Q,M}},$$
and equation \eqref{numH} follows from counting formulas \eqref{numb1} and \eqref{numb2}.
\end{proof}

The interior $\wt H_M \subset H_M$  contains only the points of $H_M$ belonging to the interior of $F_Q$,
\begin{equation*}
\wt{H}_M=\tfrac{1}{M}(P\setminus Q)\cap \mathrm{int}(F_Q)
\end{equation*}
and thus, it is formed by points from $\wt{F}_{P,M}$ which are not in $\wt{F}_{Q,M}$,
\begin{equation}\label{intHM}
\wt{H}_M=\wt{F}_{P,M}\setminus \wt{F}_{Q,M}.
\end{equation}
The explicit form of $\wt{H}_M$ is derived from \eqref{gridintF} and \eqref{intHM},
$$\wt{H}_M=\set{\tfrac{s_1}{M}\om_1+\tfrac{s_2}{M}\om_2}{s_0,s_1,s_2\in\N,s_0+s_1+s_2=M,s_1+2s_2\neq 0\,\mathrm{mod}\,3}$$
and counting formulas \eqref{numb3}, \eqref{numb4} and \eqref{intHM} yield the following proposition. 
\begin{tvr}
The number of points in the point set $\wt{H}_{M}$ for $M>3$ is given by
\begin{equation}\label{numHint}
 \abs{\wt{H}_M}=\begin{cases}
\tfrac13(M^2-3M)&M=0\;\mathrm{mod}\,3,\\
\tfrac13(M^2-3M+2)&\text{otherwise}.
\end{cases}
\end{equation}
\end{tvr} 
The point sets $H_M$ and $\wt H_M$ are for $M=6$ depicted in Figure \ref{honey}. 
\begin{figure}
\includegraphics{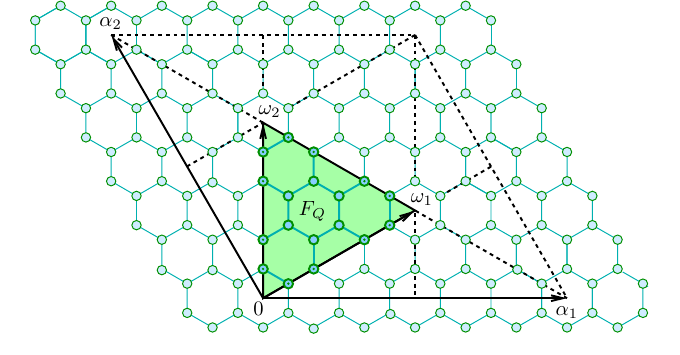}
\caption{\small The triangular fundamental domain $F_Q$ together with the vectors of the $\al-$ and $\omega-$bases are depicted. The blue nodes represent 18 points of the honeycomb lattice fragment $H_6$. Omitting the dotted nodes on the boundary of $F_Q$ yields the 6 points of the interior point set $\wt{H}_6$.}
\label{honey}
\end{figure}

\subsection{Weight sets $L_M$ and $\wt L_M$}\

For any $M\in \N$, the weight set $L_M \subset \Lambda_{P,M}$ contains the points $\la\in\Lambda_{P,M}$, which are not stabilized by~ $\Gamma_M$,
\begin{equation}\label{defLM}
L_M= \Lambda_{P,M}\setminus\Lambda^\mathrm{fix}_M.
\end{equation}
Relations \eqref{gridLa} and \eqref{defLM} yields the explicit form of $L_M$,
\begin{align*}
L_M=&\set{[\la_0,\la_1,\la_2]\in \Lambda_{Q,M}}{(\la_0>\la_1,\la_0>\la_2)\,\vee\,(\la_0=\la_1> \la_2)}.
\end{align*}

\begin{tvr}
The number of points in the weight set $L_M$ is given by
\begin{equation}\label{numL}
\abs{L_M}=\tfrac12\abs{H_M}.
\end{equation}
\end{tvr}
\begin{proof}
Formula \eqref{defLM} implies for the number of points that 
$$\abs{L_M}=\abs{\Lambda_{P,M}}-\abs{\Lambda^\mathrm{fix}_M}.$$
Using formulas \eqref{numb2}, \eqref{numbla} and \eqref{numb5}, the number of points in $L_M$ is equal to
\begin{equation}\label{LMnum}
\abs{L_M}=\begin{cases}
\tfrac16(M^2+3M)&M=0\;\mathrm{mod}\,3,\\
\tfrac{1}{6}(M^2+3M+2)& \text{otherwise}.
\end{cases}	
\end{equation}
Direct comparison of counting relations \eqref{LMnum} and \eqref{numH} guarantees \eqref{numL}.
\end{proof}

The interior $\wt{L}_M\subset \wt\Lambda_{P,M}$ contains the points $\la\in\wt\Lambda_{P,M}$, which are not stabilized by $\Gamma_M$,
\begin{equation}\label{defLM2}
\wt{L}_M= \wt\Lambda_{P,M}\setminus\Lambda^\mathrm{fix}_M.
\end{equation}
The explicit form of $\wt{L}_M$ is derived from \eqref{gridintLa} and \eqref{defLM2},
\begin{align*}
\wt L_M=&\set{[\la_0,\la_1,\la_2]\in \wt\Lambda_{Q,M}}{(\la_0>\la_1,\la_0>\la_2)\,\vee\,(\la_0=\la_1> \la_2)}
\end{align*}
and formulas \eqref{numb4}, \eqref{numb5}, \eqref{numbintla} and \eqref{numHint} yield the following proposition.
\begin{tvr}
The number of points in the grid $\wt{L}_M$ for $M>3$ is given by
\begin{equation*}
\abs{\wt{L}_M}=\tfrac12\abs{\wt{H}_M}.
\end{equation*}
\end{tvr}

The weight sets $L_M$ and $\wt L_M$ are for $M=6$ depicted in Figure \ref{obrLM}. 
\begin{figure}
\includegraphics{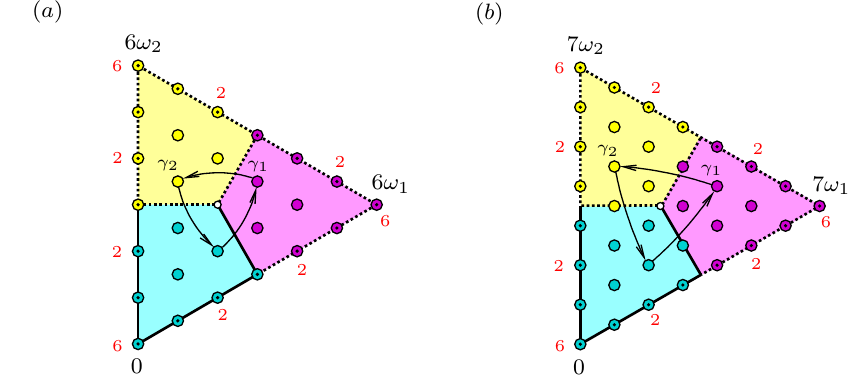}
\caption{\small $(a)$ The weight set $L_6$ consists of 9 cyan nodes. Omitting the dotted nodes on the boundary of $6F_Q$, yields 3 points of the interior weight $\wt{L}_6$. Action of the group $\Gamma_6$ is illustrated on the weight $\la=[3,2,1]\in L_6$. $(b)$ The weight set $L_7$ consists of 12 cyan nodes. Omitting the dotted nodes on the boundary of $7F_Q$, yields 5 points of the interior weight $\wt{L}_7$. Action of the group $\Gamma$ is illustrated on the weight $\la=[4,2,1]\in L_7$.  }
\label{obrLM}
\end{figure}

\section{Weyl orbit functions}

\subsection{$C-$ and $S-$functions}\

Two families of complex-valued smooth functions of variable $x\in\R^2$, labelled by $b\in P$, are defined via one-variable exponential functions as 
\begin{align}\label{orbfunC}
\Phi_b(x)&=\sum_{w\in W} e^{2 \pi i \sca{ wb}{x}},\\
\phi_b(x)&=\sum_{w\in W}\det (w)\, e^{2 \pi i \sca{ wb}{x}}.\label{orbfunS}
\end{align}
Properties of the Weyl orbit functions have been extensively studied in several articles \cite{KP1,KP2}. The functions \eqref{orbfunC} and \eqref{orbfunS} are called $C-$ and $S-$functions, respectively.
Explicit formulas for $C-$ and $S-$functions, with a weight $b=b_1\om_1+b_2\om_2$ and a point $x=x_1\omega_1+x_2\omega_2$  in $\om-$basis, are derived by employing scalar product formula \eqref{scprod} and Weyl orbit expression \eqref{orbit}, 
\begin{align}
\Phi_b(x)=& e^{\tfrac23\pi\i((2b_1+b_2)x_1+(b_1+2b_2)x_2)}+e^{\tfrac23\pi\i(-b_1+b_2)x_1+(b_1+2b_2)x_2)}+e^{\tfrac23\pi\i((-b_1-2b_2)x_1+(b_1-b_2)x_2)}\nonumber \\&+e^{\tfrac23\pi\i((-b_1-2b_2)x_1+(-2b_1-b_2)x_2)}+e^{\tfrac23\pi\i((-b_1+b_2)x_1+(-2b_1-b_2)x_2)}+e^{\tfrac23\pi\i((2b_1+b_2)x_1+(b_1-b_2)x_2)},\label{Cexp}\\
\phi_b(x)=& e^{\tfrac23\pi\i((2b_1+b_2)x_1+(b_1+2b_2)x_2)}-e^{\tfrac23\pi\i(-b_1+b_2)x_1+(b_1+2b_2)x_2)}+e^{\tfrac23\pi\i((-b_1-2b_2)x_1+(b_1-b_2)x_2)}\nonumber\\&-e^{\tfrac23\pi\i((-b_1-2b_2)x_1+(-2b_1-b_2)x_2)}+e^{\tfrac23\pi\i((-b_1+b_2)x_1+(-2b_1-b_2)x_2)}-e^{\tfrac23\pi\i((2b_1+b_2)x_1+(b_1-b_2)x_2)}.\nonumber
\end{align}

Recall from \cite{KP1,KP2} that $C-$ and $S-$functions are (anti)symmetric with respect to the Weyl group, i.e. for any $w\in W$ it holds that
\begin{alignat}{4}
&\Phi_b(wx)&=\Phi_b(x),&\q\q& \phi_b(wx)=\det (w) \phi_b(x),\label{symW}
\end{alignat}
Furthermore, both families are invariant with respect to translations by any $q\in Q$,
\begin{equation}
\Phi_b(x+q)=\Phi_b(x),\q\q \phi_b(x+q)=\phi_b(x).\label{symQ}
\end{equation}
Relations \eqref{symW} and \eqref{symQ} imply that the Weyl orbit functions are (anti)symmetric with respect to the affine Weyl group and, thus, they are restricted only to the fundamental domain \eqref{F} of the affine Weyl group. Moreover, the $S-$functions vanish on the boundary of $F_Q$ and the normal derivative of the $C-$functions to the boundary of $F_Q$ is zero. 

Denoting the Hartley kernel function by 
\begin{equation}\label{cas}
	\mathrm{cas}\,\alpha=\cos{\alpha}+\sin{\alpha}, \q \al\in \R,
\end{equation}
the Weyl orbit functions are modified \cite{HMdis,HJedis} as
\begin{align}\label{harfunC}
\hart^\id_b(x)&=\sum_{w\in W}\mathrm{cas}\,(2\pi\langle wb,x\rangle),\\
\hart^e_b(x)&=\sum_{w\in W}\det(w)\mathrm{cas}\,(2\pi\langle wb,x\rangle).\label{harfunS}
\end{align}
The (anti)symmetry relation \eqref{symW} and $Q-$shift invariance \eqref{symQ} are preserved by the Hartley functions,
\begin{align}
\hart^\id_b(wx)&=\hart^\id_b(x),\q\q\hart^e_b(wx)=\det (w) \hart^e_b(x),\label{Hartarg}\\
\hart^\id_b(x+q)&=\hart^\id_b(x),\q\q \hart^e_b(x+q)=\hart^e_b(x).\label{Hartshift}
\end{align}
Therefore, the Hartley $S-$functions $\hart^e_b$  vanish on the boundary of $F_Q$ and the normal derivative of the Hartley $C-$functions $\hart^\id_b$ to the boundary of $F_Q$ is also zero.

\subsection{Discrete orthogonality on $F_{P,M}$ and $\wt{F}_{P,M}$}\

Using coefficients $\ep(s)$ from Table \ref{eps}, a scalar product of two functions $f,g:F_{P,M}\map \C$ on the refined fragment of the weight lattice \eqref{FPM} is defined  as 
\begin{equation}\label{scaC}
\sca{f}{g}_{F_{P,M}}=\sum_{s\in F_{P,M}}\ep(s)f(s)\overline{g(s)}.
\end{equation}
Discrete orthogonality relations of the $C-$functions \eqref{orbfunC} and Hartley $C-$functions \eqref{harfunC}, labelled by the weights from the weight set \eqref{LQM} and with respect to the scalar product \eqref{scaC}, are derived in \cite{HP,HJedis}. The discrete orthogonality relations are for any $\la,\la'\in\Lambda_{Q,M}$ of the form
\begin{align}\label{ort1}
\sca{\Phi_\la}{\Phi_{\la'}}_{F_{P,M}} &=18M^2h_M(\la)\delta_{\la\la'}, \\
\sca{\hart^\id_\la}{\hart^\id_{\la'}}_{F_{P,M}}&=18M^2h_M(\la)\delta_{\la\la'},\nonumber
\end{align}
where the coefficients $h_M(\la)$ are listed in Table \ref{hM}.

Since for any $M\in\N,M>3$ and any $s\in \wt{F}_{P,M}$ it holds that $\ep(s)=6$, a scalar product of two complex valued functions $f,g :\wt{F}_{P,M}\map \C$ on the interior of the refined fragment of the weight lattice~\eqref{FPMint} is defined as
\begin{equation}\label{scaS}
\sca{f}{g}_{\wt{F}_{P,M}}=6\sum_{s\in \wt{F}_{P,M}}f(s)\overline{g(s)}.
\end{equation}

Discrete orthogonality relations of the $S-$functions \eqref{orbfunS} and Hartley $S-$functions \eqref{harfunS}, labelled by the weights from the interior weight set \eqref{LQMint} and with respect to the scalar product \eqref{scaS} are derived in \cite{HP,HJedis}. The discrete orthogonality relations are for any $\la,\la'\in\wt\Lambda_{Q,M}$ of the form
\begin{align*}
\sca{\phi_\la}{\phi_{\la'}}_{\wt{F}_{P,M}}=18M^2\delta_{\la\la'}, \\
\sca{\hart^e_\la}{\hart^e_{\la'}}_{\wt{F}_{P,M}}=18M^2\delta_{\la\la'}.
\end{align*}

\subsection{Discrete orthogonality on $F_{Q,M}$ and $\wt{F}_{Q,M}$}\

A scalar product of two functions $f,g :F_{Q,M}\map \C$ on the refined fragment of the root lattice \eqref{FQM} is defined as
\begin{equation}\label{scaQC}
\sca{f}{g}_{F_{Q,M}}=\sum_{s\in F_{Q,M}}\ep(s)f(s)\overline{g(s)}.
\end{equation}
Discrete orthogonality relations of the $C-$functions \eqref{orbfunC} and Hartley $C-$functions \eqref{harfunC}, labelled by the weights from the weight set \eqref{LPM} and with respect to the scalar product \eqref{scaQC}, are derived in \cite{HMdis}. For any $\la,\la'\in\Lambda_{P,M}$ it holds that
\begin{align}\label{ort3}
\sca{\Phi_\la}{\Phi_{\la'}}_{F_{Q,M}}&=6M^2d(\la)h_M(\la)\delta_{\la\la'},\\
\sca{\hart^\id_\la}{\hart^\id_{\la'}}_{F_{Q,M}} &= 6M^2d(\la)h_M(\la)\delta_{\la\la'},\nonumber
\end{align}
where $$d(\la)=\begin{cases}
3&\la_0=\la_1=\la_2,\\
1&\text{otherwise}.
\end{cases}
$$

A scalar product of two functions $f,g :\wt{F}_{Q,M}\map \C$ on the interior of the refined fragment of the root lattice~\eqref{FQMint} is defined as
\begin{equation}\label{scaQS}
\sca{f}{g}_{\wt{F}_{Q,M}}=6\sum_{s\in \wt{F}_{Q,M}}f(s)\overline{g(s)}.
\end{equation}
Discrete orthogonality relations of the $S-$functions \eqref{orbfunS} and Hartley $S-$functions \eqref{harfunS}, labelled by the weights from the weight set \eqref{LPMint} and with respect to the scalar product \eqref{scaQS}, are derived in \cite{HMdis}. For any $\la,\la'\in\wt\Lambda_{P,M}$ it holds that,
\begin{align*}
\sca{\phi_\la}{\phi_{\la'}}_{\wt{F}_{Q,M}} &= 6M^2d(\la)\delta_{\la\la'}, \\
\sca{\hart^e_\la}{\hart^e_{\la'}}_{\wt{F}_{Q,M}} &=6M^2d(\la)\delta_{\la\la'}.
\end{align*}

\section{Honeycomb Weyl and Hartley orbit functions}

\subsection{Extended $C-$ and $S-$functions}\

Extended Weyl orbit functions are complex valued smooth functions  induced from the standard $C-$ and $S-$functions. For a fixed $M\in\N$, the extended  $C-$functions $\Phi^\pm_\la$ of variable $x\in \R^2$, labeled by $\la\in L_M$, are introduced by
\begin{equation}\label{def1}
\begin{aligned} \Phi^+_\la (x)&=  \mu^{+,0}_{\la} \Phi_{\la}(x)+ \mu^{+,1}_{\la} \Phi_{\gamma_1\la}(x)+\mu^{+,2}_{\la} \Phi_{\gamma_2\la}(x), \\
 \Phi^-_\la (x)&=  \mu^{-,0}_{\la} \Phi_{\la}(x)+ \mu^{-,1}_{\la} \Phi_{\gamma_1\la}(x) +\mu^{-,2}_{\la} \Phi_{\gamma_2\la}(x),
\end{aligned}
\end{equation}
where $\mu_\la^{\pm,0},\mu_\la^{\pm,1},\mu_\la^{\pm,2}\in \C$ denote for each $\la\in L_M$ six arbitrary extension coefficients. For a fixed $M>3$, the extended $S-$functions $\phi^\pm_\la$ of variable $x\in\R^2$, labeled by $\la\in \wt{L}_M$, are introduced by
\begin{equation}\label{def2}
\begin{aligned} \phi^+_\la (x)&=  \mu^{+,0}_{\la} \phi_{\la}(x)+ \mu^{+,1}_{\la} \phi_{\gamma_1\la}(x)+\mu^{+,2}_{\la} \phi_{\gamma_2\la}(x), \\
 \phi^-_\la (x)&=  \mu^{-,0}_{\la} \phi_{\la}(x)+ \mu^{-,1}_{\la} \phi_{\gamma_1\la}(x) +\mu^{-,2}_{\la} \phi_{\gamma_2\la}(x).
\end{aligned}
\end{equation}

The extended $C-$ and $S-$functions inherit the argument symmetry \eqref{symW} of Weyl orbit functions with respect to any $w\in W$,
\begin{equation}\label{prop1}
\Phi^\pm_\la(wx)=\Phi^\pm_\la(x),\q \phi^\pm_\la(wx)=\det (w)\phi^\pm_\la(x),
\end{equation}
and the invariance \eqref{symQ} with respect to the shifts from $q\in Q$, 
\begin{equation}\label{prop2}
\Phi^\pm_\la(x+q)=\Phi^\pm_\la(x),\q \phi^\pm_\la(x+q)=\phi^\pm_\la(x).
\end{equation}
Relations \eqref{prop1} and \eqref{prop2} imply that the extended Weyl orbit functions are also (anti)symmetric with respect to the affine Weyl group and, thus, they are restricted only to the fundamental domain \eqref{F} of the affine Weyl group. Moreover, the extended $S-$functions vanish on the boundary of $F_Q$ and the normal derivative of the extended $C-$functions to the boundary of $F_Q$ is zero.

The extended Weyl orbit functions \eqref{def1} and \eqref{def2} are modified using the Hartley orbit functions \eqref{harfunC} and \eqref{harfunS}. For a fixed $M\in\N$, the extended Hartley $C-$functions $\hart^{\id,\pm}_\la$ of variable $x\in\R^2$, parametrized by $\la\in L_M$, are defined by
\begin{equation}\label{extHarC}
\begin{aligned} \hart^{\id,+}_\la (x)&=  \mu^{+,0}_{\la} \hart^\id_{\la}(x)+ \mu^{+,1}_{\la} \hart^\id_{\gamma_1\la}(x)+\mu^{+,2}_{\la} \hart^\id_{\gamma_2\la}(x), \\
 \hart^{\id,-}_\la (x)&=  \mu^{-,0}_{\la} \hart^\id_{\la}(x)+ \mu^{-,1}_{\la} \hart^\id_{\gamma_1\la}(x) +\mu^{-,2}_{\la} \hart^\id_{\gamma_2\la}(x),
\end{aligned}
\end{equation}
where $\mu_\la^{\pm,0},\mu_\la^{\pm,1},\mu_\la^{\pm,2}\in \C$ denote for each $\la\in L_M$ six arbitrary extension coefficients. For a fixed $M>3$, the extended Hartley $S-$functions $\hart^{e,\pm}_\la$ of variable $x\in\R^2$, labelled by $\la\in \wt{L}_M$, are introduced by
\begin{equation}\label{extHarS}
\begin{aligned} \hart^{e,+}_\la (x)&=  \mu^{+,0}_{\la} \hart^e_{\la}(x)+ \mu^{+,1}_{\la} \hart^e_{\gamma_1\la}(x)+\mu^{+,2}_{\la} \hart^e_{\gamma_2\la}(x), \\
 \hart^{e,-}_\la (x)&=  \mu^{-,0}_{\la} \hart^e_{\la}(x)+ \mu^{-,1}_{\la} \hart^e_{\gamma_1\la}(x) +\mu^{-,2}_{\la} \hart^e_{\gamma_2\la}(x).
\end{aligned}
\end{equation}
Restricting the extension coefficients to real numbers $\mu_\la^{\pm,0},\mu_\la^{\pm,1},\mu_\la^{\pm,2}\in \R$, the functions $\hart^{\id,\pm}_\la$ and  $\hart^{e,\pm}_\la$ become real valued.
Similarly to the extended $C-$ and $S-$functions, the extended Hartley $C-$ and $S-$functions inherit argument symmetry \eqref{Hartarg} of the $\hart^\id-$functions and $\hart^e-$functions with respect to any $w\in W,$ 
\begin{equation}\label{prop3}
\hart^{\id,\pm}_\la(wx)=\hart^{\id,\pm}_\la(x),\q \hart^{e,\pm}_\la(wx)=\det (w)\hart^{e,\pm}_\la(x),
\end{equation}
and the invariance \eqref{Hartshift} with respect to the shifts from $q\in Q$, 
\begin{equation}\label{prop4}
\hart^{\id,\pm}_\la(x+q)=\hart^{\id,\pm}_\la(x),\q \hart^{e,\pm}_\la(x+q)=\hart^{e,\pm}_\la(x).
\end{equation}
Relations \eqref{prop3} and \eqref{prop4} imply that the extended Hartley orbit functions are also (anti)symmetric with respect to the affine Weyl group and, thus, they are restricted only to the fundamental domain \eqref{F} of the affine Weyl group. Moreover, the extended Hartley $S-$functions vanish on the boundary of $F_Q$ and the normal derivative of the extended Hartley $C-$functions to the boundary of $F_Q$ is zero.

\subsection{Honeycomb $C-$ and $S-$ functions}\

Special classes of extended $C-$ and $S-$functions and their Hartley versions are obtained by imposing three additional conditions on the extension coefficients $\mu_\la^{\pm,0},\mu_\la^{\pm,1},\mu_\la^{\pm,2}\in \C$. 
Two discrete normalization functions $\mu^+,\mu^- :L_M  \map \R $  are for any $\la \in L_M$ defined as
\begin{equation}\label{mu}
\mu^\pm(\la)=\abs{\mu^{\pm,0}_\la}^2+\abs{\mu^{\pm,1}_\la}^2+\abs{\mu^{\pm,2}_\la}^2-\mathrm{Re}{\left(\mu^{\pm,0}_\la\overline{\mu^{\pm,1}_\la}+\mu^{\pm,0}_\la\overline{\mu^{\pm,2}_\la}+\mu^{\pm,1}_\la\overline{\mu^{\pm,2}_\la}\right)},
\end{equation}
and an intertwining function $\beta :L_M  \map \C $ is defined as 
\begin{equation}\label{beta}
\begin{alignedat}{2}
&\beta(\lambda)&=&\, 2\left(\mu^{+,0}_\la\overline{\mu^{-,0}_{\la}}+\mu^{+,1}_\la\overline{\mu^{-,1}_{\la}}+\mu^{+,2}_\la\overline{\mu^{-,2}_{\la}}\right)-\mu^{+,0}_\la\left(\overline{\mu^{-,1}_{\la}}+\overline{\mu^{-,2}_{\la}}\right)\\  &&&-\mu^{+,1}_\la\left(\overline{\mu^{-,0}_{\la}}+\overline{\mu^{-,2}_{\la}}\right)-\mu^{+,2}_\la\left(\overline{\mu^{-,0}_{\la}}+\overline{\mu^{-,1}_{\la}}\right).
\end{alignedat}
\end{equation}

The $\Phi^\pm_\la-$functions \eqref{def1}, for which both normalization functions are positive and the intertwining functions vanishes,
\begin{equation}\label{cond1}
\mu^\pm(\lambda)>0,\q\beta(\lambda)=0,\q \la \in L_M,
\end{equation}
are named the honeycomb $C-$functions and denoted by $\mathrm{Ch}^\pm_\la$. Similarly, the $\phi^\pm_\la-$functions \eqref{def2}, which satisfy 
\begin{equation}\label{cond2}
\mu^\pm(\lambda)>0,\q \beta(\lambda)=0,\q \la \in \wt L_M, 
\end{equation}
 are named the honeycomb $S-$functions and denoted by $\mathrm{Sh}^\pm_\la$.
The extended Hartley $C-$functions \eqref{extHarC} satisfying the conditions \eqref{cond1} are named the honeycomb Hartley $C-$functions and denoted by $\mathrm{Cah}^\pm_\la$. The extended Hartley $S-$functions \eqref{extHarS} satisfying the conditions \eqref{cond2} are named the honeycomb Hartley $S-$functions and denoted by $\mathrm{Sah}^\pm_\la$.

For any two complex discrete functions $f,g:H_M\map \C$, a scalar product on the finite fragment of the honeycomb lattice \eqref{HM} is defined as
\begin{equation}\label{sca1}
\langle f,g\rangle_{H_M}=\sum_{s\in H_M}\ep(s)f(s)\overline{g(s)},
\end{equation}
and the resulting finite-dimensional Hilbert space of complex valued functions is denoted by $\mathcal{H}_M$.
\begin{thm}\label{thmdis}
Any set of the honeycomb $C-$functions $\mathrm{Ch}_\la^\pm$, $\la\in L_M$, restricted to $H_M$, forms an orthogonal basis of the space $\mathcal{H}_M$. For any $\la,\la'\in L_M$ it holds that
\begin{align}
&\langle \mathrm{Ch}_\la^\pm,\mathrm{Ch}_{\la'}^{\pm} \rangle_{H_M}=12M^2h_M(\la)\mu^\pm(\la)\delta_{\la\la'}, \label{hdis1}\\
&\langle \mathrm{Ch}_\la^+,\mathrm{Ch}_{\la'}^{-} \rangle_{H_M}=0.\label{hdis2}
\end{align}
\end{thm}

\begin{proof}
Let $t$ and $t'$ stand for the symbols $+$ and $-$, i.e. $t,t'\in \{+,-\}$ . The point set relation \eqref{HM} guarantees for the scalar products \eqref{scaC}, \eqref{scaQC} and \eqref{sca1} that
\begin{equation}\label{scap}
	\langle \mathrm{Ch}^t_\la,\mathrm{Ch}^{t'}_{\la'}\rangle_{H_M}=\langle \mathrm{Ch}^t_\la,\mathrm{Ch}^{t'}_{\la'}\rangle_{F_{P,M}}-\langle \mathrm{Ch}^t_\la,\mathrm{Ch}^{t'}_{\la'}\rangle_{F_{Q,M}}.
\end{equation}
Substituting definition of the extended $C-$functions \eqref{def1} into \eqref{scap} yields
\begin{equation}\label{simp}
		\langle \mathrm{Ch}^t_\la,\mathrm{Ch}^{t'}_{\la'}\rangle_{H_M}=\sum_{k,l=0}^2\mu_\la^{t,k}\overline{\mu_{\la'}^{t',l}}\left(\langle\Phi_{\gamma_k\la},\Phi_{\gamma_l\la'}\rangle_{F_{P,M}}-\langle\Phi_{\gamma_k\la},\Phi_{\gamma_l\la'}\rangle_{F_{Q,M}}\right)
\end{equation}

Relation \eqref{gam} and definition \eqref{defLM} grant that both $\gamma_k\la,\gamma_l\la'\in\Lambda_{Q,M}$ and thus, the   
discrete orthogonality relations \eqref{ort1} and the $\Gamma_M-$invariance \eqref{invgam} ensure that 
\begin{equation}\label{pom0}
\langle\Phi_{\gamma_k\la},\Phi_{\gamma_l\la'}\rangle_{F_{P,M}}=18M^2h_M(\la)\delta_{\gamma_k\la,\gamma_l\la'}.
\end{equation}
Since $F_P$ is a fundamental domain of $W_P^{\mathrm{aff}}$, the equality $\gamma_k\la=\gamma_l\la'$ and definition \eqref{LPM} imply that $\la=\la'$ and $\gamma_k^{-1}\gamma_l$ stabilizes $\la \in L_M$. Then since  $\Gamma_M$ is a cyclic group of prime order and the stabilizer subgroup of $\la \in L_M$ cannot be due to \eqref{defLM} the entire $\Gamma_M$, it follows that  $\gamma_k=\gamma_l$. Thus, the orthogonality relation \eqref{pom0} is simplified as
\begin{equation}\label{pom1}
\langle\Phi_{\gamma_k\la},\Phi_{\gamma_l\la'}\rangle_{F_{P,M}}=18M^2h_M(\la)\delta_{\la\la'}\delta_{kl}.
\end{equation}

The explicit form \eqref{gammaexp} of the group $\Gamma_M$ yields for any $w\in W$ and $s\in\tfrac{1}{M}Q$ the equality 
\begin{equation}\label{scash}
\langle w\gamma_k\la,s\rangle=\langle w(r_1r_2)^k\la+Mw\om_k,s\rangle=\langle w(r_1r_2)^k\la,s\rangle+\langle Mw\om_k,s\rangle, \q k=1,2	.
\end{equation}
Since the lattices $P$ and $Q$ are $W-$invariant and $\Z-$dual due to relation \eqref{Zdual}, it holds for any $s\in F_{Q,M}$ that  $\langle Mw\om_k,s\rangle \in \Z$ and thus, for all $\gamma_k \in \Gamma_M$ the following identity is obtained,
\begin{equation}\label{iddd}
\Phi_{\gamma_k\la}(s)=\Phi_{\la}(s),\q s\in F_{Q,M}.	
\end{equation}
The discrete orthogonality relations \eqref{ort3} and \eqref{iddd} then grant that
\begin{equation}\label{pom2}
\langle\Phi_{\gamma_k\la},\Phi_{\gamma_l\la'}\rangle_{F_{Q,M}}=\langle\Phi_{\la},\Phi_{\la'}\rangle_{F_{Q,M}}=6M^2h_M(\la)\delta_{\la\la'}.
\end{equation}
Substituting the resulting scalar products \eqref{pom1} and \eqref{pom2} into \eqref{simp} produces the relations
\begin{align}
\langle \mathrm{Ch}^t_\la,\mathrm{Ch}^{t'}_{\la'}\rangle_{H_M}&=6M^2h_M(\la)\left(3\sum_{k=0}^2\mu_\la^{t,k}\overline{\mu_{\la'}^{t',k}}-\sum_{k,l=0}^2\mu_\la^{t,k}\overline{\mu_{\la'}^{t',l}}\right)\delta_{\la\la'} \nonumber \\
&=\begin{cases}
12M^2h_M(\la)\mu^t(\la)\delta_{\la\la'},&t=t',\\
6M^2h_M(\la)\beta(\la)\delta_{\la\la'},&t\neq t'.
\end{cases} \label{res}
\end{align}
Conditions \eqref{cond1} for the honeycomb $C-$functions $\mathrm{Ch}^\pm_\la$, $\la\in L_M$ and \eqref{res}  then guarantee the discrete orthogonality relations \eqref{hdis1} and \eqref{hdis2}. According to counting relation \eqref{numL}, the number of the orthogonal functions $\mathrm{Ch}^\pm_\la$, $\la\in L_M$ coincides with the cardinality of $H_M$ and therefore,
these functions form an orthogonal basis of $\mathcal{H}_M$. 
\end{proof}

Note that the result of the scalar product \eqref{res} grants that the normalization functions \eqref{mu} are for any extension coefficients $\mu_\la^{\pm,0},\mu_\la^{\pm,1},\mu_\la^{\pm,2}\in \C$ always non-negative    
\begin{equation}\label{pos}
\mu^\pm(\lambda)\geq 0,\q \la \in L_M.
\end{equation}
Similarly, the following Hartley version of Theorem \ref{thmdis} is deduced.
\begin{thm}\label{thmdish}
Any set of the honeycomb Hartley $C-$functions $\mathrm{Cah}_\la^\pm$, $\la\in L_M$, restricted to $H_M$, forms an orthogonal basis of the space $\mathcal{H}_M$. For any $\la,\la'\in L_M$ it holds that
\begin{align*}
&\langle \mathrm{Cah}_\la^\pm,\mathrm{Cah}_{\la'}^{\pm} \rangle_{H_M}=12M^2h_M(\la)\mu^\pm(\la)\delta_{\la\la'},\\
&\langle \mathrm{Cah}_\la^+,\mathrm{Cah}_{\la'}^{-} \rangle_{H_M}=0.
\end{align*}
\end{thm}

For any two complex discrete functions $f,g:\wt H_M\map \C$, a scalar product on the interior fragment of the honeycomb lattice \eqref{intHM} is defined as
\begin{equation*}
\langle f,g\rangle_{\wt{H}_M}=6\sum_{s\in \wt{H}_M}f(s)\overline{g(s)},
\end{equation*}
and the resulting finite-dimensional Hilbert space of complex valued functions is denoted by $\wt{\mathcal{H}}_M$. As in Theorems \ref{thmdis} and \ref{thmdish}, the discrete orthogonality of the honeycomb $
S-$functions and the honeycomb Hartley $
S-$functions is obtained.
\begin{thm}\label{thmdis2}
Any set of the honeycomb $S-$functions $\mathrm{Sh}_\la^\pm$, $\la\in \wt L_M$, restricted to $\wt H_M$, forms an orthogonal basis of the space $\wt{\mathcal{H}}_M$. Any set of the honeycomb Hartley $S-$functions $\mathrm{Sah}_\la^\pm$, $\la\in \wt L_M$, restricted to $\wt H_M$, forms an orthogonal basis of the space $\wt{\mathcal{H}}_M$. For any $\la,\la'\in \wt L_M$ it holds that
\begin{align*}
&\langle \mathrm{Sh}_\la^\pm,\mathrm{Sh}_{\la'}^\pm \rangle_{\wt{H}_M}=\langle \mathrm{Sah}_\la^\pm,\mathrm{Sah}_{\la'}^\pm \rangle_{\wt{H}_M}=12M^2\mu^\pm(\la)\delta_{\la\la'},\\
&\langle \mathrm{Sh}_\la^+,\mathrm{Sh}_{\la'}^- \rangle_{\wt{H}_M}=\langle \mathrm{Sah}_\la^+,\mathrm{Sah}_{\la'}^- \rangle_{\wt{H}_M}=0.
\end{align*}
\end{thm}

\section{Three types of honeycomb $C-$ and $S-$functions}

\subsection{Type $\mathrm{I}$}\

The first type of the honeycomb $C-$ and $S-$functions and their Hartley versions is characterized by real common values, independent of $\la\in L_M$, of the extension coefficients $\mu^{\pm,k}_\la$.
One of the simplest choices of the values $\mu^{\pm,k}_\la$ satisfying the conditions \eqref{cond1} and \eqref{cond2} is
\begin{equation}\label{typeI}
\begin{aligned}
&\left(\mu^{+,0}_\la,\mu^{+,1}_\la,\mu^{+,2}_\la\right)=(1,0,0),\\
&\left(\mu^{-,0}_\la,\mu^{-,1}_\la,\mu^{-,2}_\la\right)=(0,1,-1).
\end{aligned}
\end{equation}
The intertwining function \eqref{beta} indeed vanishes and the normalization functions \eqref{mu} have constant values 
$$\mu^+(\la)=1,\,\q\mu^-(\la)=3, \q \la \in L_M.$$

Given $\la=[\la_0,\la_1,\la_2]=\la_1\om_1+\la_2\om_2\in L_M$ and $x=x_1\om_1+x_2\om_2$ in the $\om-$basis, the honeycomb $C-$functions $\mathrm{Ch}^{\pm,\mathrm{I}}_\la$ are explicitly given by
\begin{align*}
\mathrm{Ch}^{+,\mathrm{I}}_\la(x)&=e^{\tfrac23\pi\i((2\la_1+\la_2)x_1+(\la_1+2\la_2)x_2)}+e^{\tfrac23\pi\i(-\la_1+\la_2)x_1+(\la_1+2\la_2)x_2)}\\&+e^{\tfrac23\pi\i((-\la_1-2\la_2)x_1+(\la_1-\la_2)x_2)}+e^{\tfrac23\pi\i((-\la_1-2\la_2)x_1+(-2\la_1-\la_2)x_2)}\\&+e^{\tfrac23\pi\i((-\la_1+\la_2)x_1+(-2\la_1-\la_2)x_2)}+e^{\tfrac23\pi\i((2\la_1+\la_2)x_1+(\la_1-\la_2)x_2)},
\end{align*}
\begin{align*}
\mathrm{Ch}^{-,\mathrm{I}}_\la(x)&=e^{\tfrac23\pi\i((2M-\la_1-2\la_2)x_1+(M+\la_1-\la_2)x_2)}+e^{\tfrac23\pi\i(-M+2\la_1+\la_2)x_1+(M+\la_1-\la_2)x_2)}\\&+e^{\tfrac23\pi\i((-M-\la_1+\la_2)x_1+(M-2\la_1-\la_2)x_2)}+e^{\tfrac23\pi\i((-M-\la_1+\la_2)x_1+(-2M+\la_1+2\la_2)x_2)}\\&+e^{\tfrac23\pi\i((-M+2\la_1+\la_2)x_1+(-2M+\la_1+2\la_2)x_2)}+e^{\tfrac23\pi\i((2M-\la_1-2\la_2)x_1+(M-2\la_1-\la_2)x_2)}\\
&-e^{\tfrac23\pi\i((M-\la_1+\la_2)x_1+(2M-2\la_1-\la_2)x_2)}-e^{\tfrac23\pi\i(M-\la_1-2\la_2)x_1+(2M-2\la_1-\la_2)x_2)}\\&-e^{\tfrac23\pi\i((-2M+2\la_1+\la_2)x_1+(-M+\la_1+2\la_2)x_2)}-e^{\tfrac23\pi\i((-2M+2\la_1+\la_2)x_1+(-M+\la_1-\la_2)x_2)}\\&-e^{\tfrac23\pi\i((M-\la_1-2\la_2)x_1+(-M+\la_1-\la_2)x_2)}-e^{\tfrac23\pi\i((M-\la_1+\la_2)x_1+(-M+\la_1+2\la_2)x_2)},
\end{align*}
and the honeycomb $S-$functions $\mathrm{Sh}^{\pm,\mathrm{I}}_\la$ are of the following form,
\begin{align*}
\mathrm{Sh}^{+,\mathrm{I}}_\la(x)&=e^{\tfrac23\pi\i((2\la_1+\la_2)x_1+(\la_1+2\la_2)x_2)}-e^{\tfrac23\pi\i(-\la_1+\la_2)x_1+(\la_1+2\la_2)x_2)}\\&+e^{\tfrac23\pi\i((-\la_1-2\la_2)x_1+(\la_1-\la_2)x_2)}-e^{\tfrac23\pi\i((-\la_1-2\la_2)x_1+(-2\la_1-\la_2)x_2)}\\&+e^{\tfrac23\pi\i((-\la_1+\la_2)x_1+(-2\la_1-\la_2)x_2)}-e^{\tfrac23\pi\i((2\la_1+\la_2)x_1+(\la_1-\la_2)x_2)},
\end{align*}
\begin{align*}
\mathrm{Sh}^{-,\mathrm{I}}_\la(x)&=e^{\tfrac23\pi\i((2M-\la_1-2\la_2)x_1+(M+\la_1-\la_2)x_2)}-e^{\tfrac23\pi\i(-M+2\la_1+\la_2)x_1+(M+\la_1-\la_2)x_2)}\\&+e^{\tfrac23\pi\i((-M-\la_1+\la_2)x_1+(M-2\la_1-\la_2)x_2)}-e^{\tfrac23\pi\i((-M-\la_1+\la_2)x_1+(-2M+\la_1+2\la_2)x_2)}\\&+e^{\tfrac23\pi\i((-M+2\la_1+\la_2)x_1+(-2M+\la_1+2\la_2)x_2)}-e^{\tfrac23\pi\i((2M-\la_1-2\la_2)x_1+(M-2\la_1-\la_2)x_2)}\\
&-e^{\tfrac23\pi\i((M-\la_1+\la_2)x_1+(2M-2\la_1-\la_2)x_2)}+e^{\tfrac23\pi\i(M-\la_1-2\la_2)x_1+(2M-2\la_1-\la_2)x_2)}\\&-e^{\tfrac23\pi\i((-2M+2\la_1+\la_2)x_1+(-M+\la_1+2\la_2)x_2)}+e^{\tfrac23\pi\i((-2M+2\la_1+\la_2)x_1+(-M+\la_1-\la_2)x_2)}\\&-e^{\tfrac23\pi\i((M-\la_1-2\la_2)x_1+(-M+\la_1-\la_2)x_2)}+e^{\tfrac23\pi\i((M-\la_1+\la_2)x_1+(-M+\la_1+2\la_2)x_2)}.
\end{align*}

Explicit formulas for the honeycomb Hartley $C-$ and $S-$functions $\mathrm{Cah}^{\pm,\mathrm{I}}_\la$ and $\mathrm{Sah}^{\pm,\mathrm{I}}_\la$ are obtained directly by replacing exponential functions with Hartley kernel functions \eqref{cas}. 
The contour plots the honeycomb Hartley functions $\mathrm{Cah}^{\pm,\mathrm{I}}_\la$ and $\mathrm{Sah}^{\pm,\mathrm{I}}_\la$ are depicted in Figures \ref{CahI} and \ref{SahI}, respectively.

\begin{figure}
\includegraphics{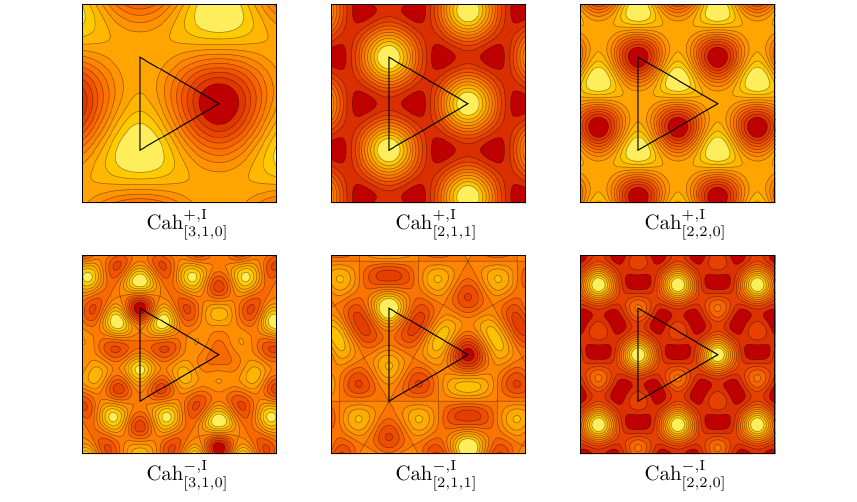}
\caption{{\small The contour plots of the honeycomb Hartley $C-$functions $\mathrm{Cah}^{\pm,\mathrm{I}}_\la$ with $M=4$. The triangle depicts the fundamental domain $F_Q$.}}
\label{CahI}
\end{figure}

\begin{figure}
\includegraphics{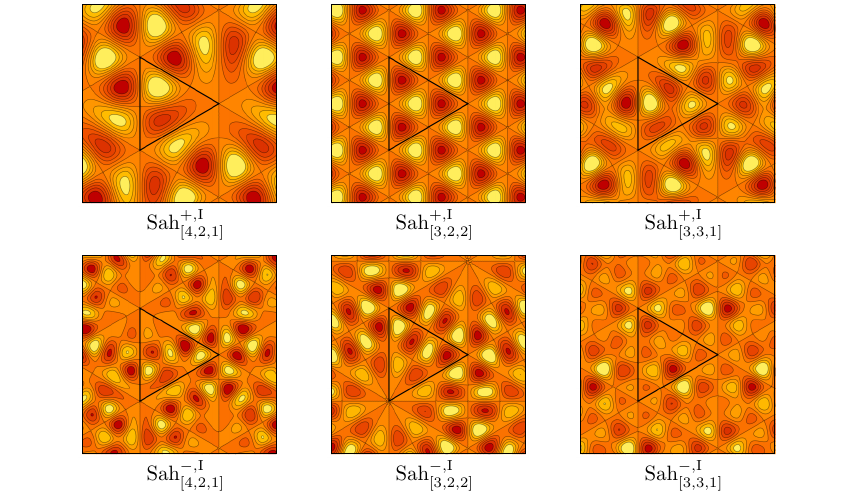}
\caption{{\small The contour plots of the honeycomb Hartley $S-$functions $\mathrm{Sah}^{\pm,\mathrm{I}}_\la$ with $M=7$. The triangle depicts the fundamental domain $F_Q$.}}
\label{SahI}
\end{figure}

\subsection{Type $\mathrm{II}$}\

The type $\mathrm{II}$ is characterized by non-constant real values of the extension coefficients $\mu^{\pm,k}_\la$. By selecting a special case, the values of $\mu^{\pm,k}_\la$ are specified as
\begin{equation}\label{typeII}
\begin{aligned}
&\mu^{\pm,0}_\la=\mathrm{Re}\left\{(3+\sqrt{3}\,\i)\Phi_\la\left(\tfrac{\om_1}{M}\right)\right\},\\
&\mu^{\pm,1}_\la=0,\\
&\mu^{\pm,2}_\la=\mathrm{Re}\left\{(3-\sqrt{3}\,\i)\Phi_\la\left(\tfrac{\om_1}{M}\right)\right\}\pm 3\abs{\Phi_\la\left(\tfrac{\om_1}{M}\right)}.
\end{aligned}
\end{equation}
The intertwining function \eqref{beta} vanishes and the normalization functions \eqref{mu} are calculated as 
\begin{align}\label{mutypeII}
\mu^\pm(\la)&=9\abs{\Phi_\la\left(\tfrac{\om_1}{M}\right)}\left(2\abs{\Phi_\la\left(\tfrac{\om_1}{M}\right)}\pm\mathrm{Re}\left\{(1-\sqrt{3}\,\i)\Phi_\la\left(\tfrac{\om_1}{M}\right)\right\}\right).
\end{align}
To verify the positivity of the normalization functions \eqref{mutypeII} in conditions \eqref{cond1} and \eqref{cond2}, the inequality $\mu^+(\la)\mu^-(\la)\neq 0$ is proven. Substituting coefficients \eqref{typeII} into defining relations \eqref{mu} yields  
\begin{equation}\label{product}	\mu^+(\la)\mu^-(\la)=81\abs{\Phi_\la\left(\tfrac{\om_1}{M}\right)}^2\left(\mathrm{Im}\left\{\Phi_\la\left(\tfrac{\om_1}{M}\right)\right\}-\sqrt{3}\mathrm{Re}\left\{\Phi_\la\left(\tfrac{\om_1}{M}\right)\right\}\right)^2.
\end{equation}
Firstly, the explicit form of $C-$functions \eqref{Cexp} produces for $\la=\la_1\om_1+\la_2\om_2$ in $\om-$basis the expression
$$|\Phi_\la\left(\tfrac{\om_1}{M}\right)|^2=4\left[2\cos{\left(\tfrac{\pi}{M}(\la_1+\la_2)\right)}+\cos{\left(\tfrac{\pi}{M}(\la_1-\la_2)\right)}\right]^2+4\sin^2{\left(\tfrac{\pi}{M}(\la_1-\la_2)\right)}.$$
Standard trigonometric identities guarantee that the system of equations
\begin{align*}
&2\cos{\left(\tfrac{\pi}{M}(\la_1+\la_2)\right)}+\cos{\left(\tfrac{\pi}{M}(\la_1-\la_2)\right)}=0,\\
&\sin{\left(\tfrac{\pi}{M}(\la_1-\la_2)\right)}=0,
\end{align*}  
has no solution for $\la\in L_M$ and hence,
$$|\Phi_\la\left(\tfrac{\om_1}{M}\right)|^2\neq0,\q \la\in L_M.$$
Secondly, the explicit form of $C-$functions \eqref{Cexp} provides the relation
\begin{align*}
\mathrm{Im}\left\{\Phi_\la\left(\tfrac{\om_1}{M}\right)\right\}-\sqrt{3}\mathrm{Re}\left\{\Phi_\la\left(\tfrac{\om_1}{M}\right)\right\}
=&-16\cos{\left(\tfrac{\pi}{3M}(2\la_1+\la_2)-\tfrac{\pi}{6}\right)}\cdot \cos{\left(\tfrac{\pi}{3M}(\la_1-\la_2)+\tfrac{\pi}{6}\right)} \\ & \cdot \cos{\left(\tfrac{\pi}{3M}(\la_1+2\la_2)+\tfrac{\pi}{6}\right)}.
\end{align*}
Since for all $\la\in L_M$ it holds that
\begin{align*}
&\cos{\left(\tfrac{\pi}{3M}(2\la_1+\la_2)-\tfrac{\pi}{6}\right)}\neq0,\\
&\cos{\left(\tfrac{\pi}{3M}(\la_1-\la_2)+\tfrac{\pi}{6}\right)}\neq0,\\
&\cos{\left(\tfrac{\pi}{3M}(\la_1+2\la_2)+\tfrac{\pi}{6}\right)}\neq0,
\end{align*}
the product \eqref{product} is non-zero and hence,
\begin{equation}\label{nonzero}
\mu^\pm(\la)\neq0,\q\la\in L_M.
\end{equation}
Positivity of the normalization functions \eqref{pos} and property \eqref{nonzero} then imply the validity of conditions \eqref{cond1} and~\eqref{cond2}. 
The contour plots the honeycomb Hartley functions $\mathrm{Cah}^{\pm,\mathrm{II}}_\la$ and $\mathrm{Sah}^{\pm,\mathrm{II}}_\la$ are depicted in Figures \ref{CahII} and \ref{SahII}, respectively.
\begin{figure}
\includegraphics{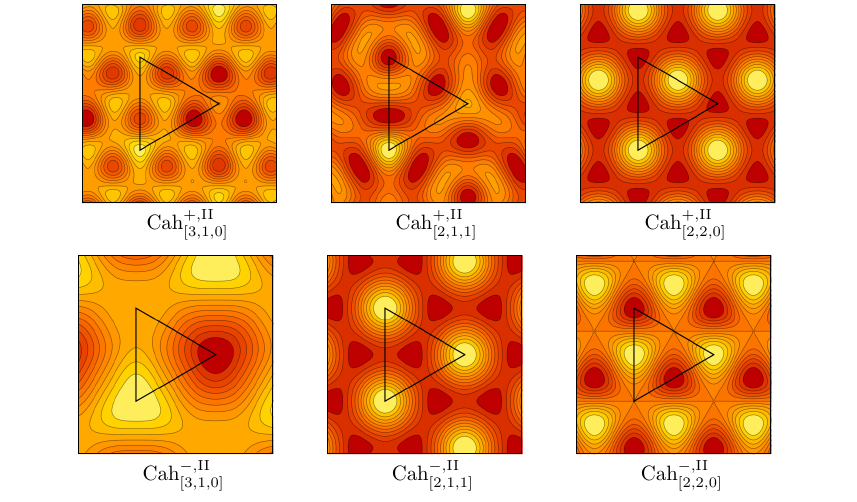}
\caption{{\small The contour plots of the honeycomb Hartley $C-$functions $\mathrm{Cah}^{\pm,\mathrm{II}}_\la$ with $M=4$. The triangle depicts the fundamental domain $F_Q$.}}
\label{CahII}
\end{figure}

\begin{figure}
\includegraphics{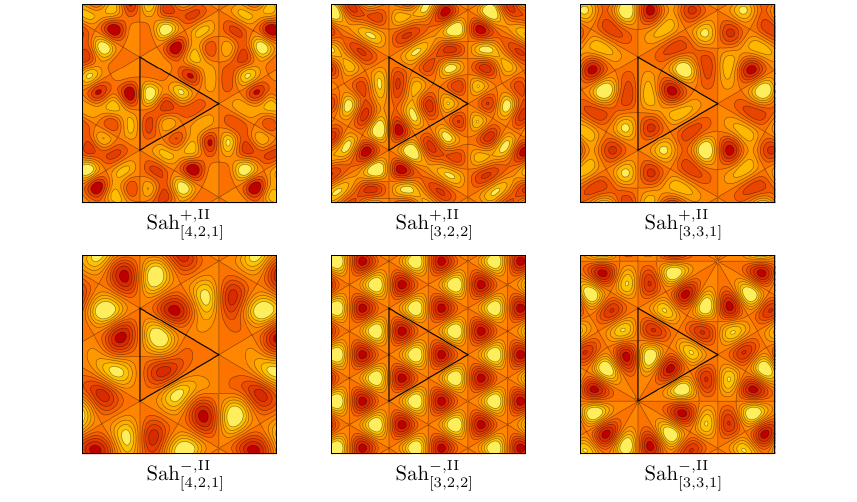}
\caption{{\small The contour plots of the honeycomb Hartley $S-$functions $\mathrm{Sah}^{\pm,\mathrm{II}}_\la$ with $M=7$. The triangle depicts the fundamental domain $F_Q$.}}
\label{SahII}
\end{figure}

\subsection{Type $\mathrm{III}$}\

The third type of the honeycomb $C-$ and $S-$functions is characterized by a common value of the extension coefficients $\mu^{\pm,k}_\la,$ independent of $\la\in L_M$ and with non-zero imaginary part.
One of the simplest choices of the complex values $\mu^{\pm,k}_\la$ of extension constants satisfying the conditions \eqref{cond1} and \eqref{cond2} is
\begin{equation}\label{mu3}
\begin{aligned}
&\left(\mu^{+,0}_\la,\mu^{+,1}_\la,\mu^{+,2}_\la\right)=\left(1,e^{\tfrac{2\pi\i}{3}},e^{-\tfrac{2\pi\i}{3}}\right),\\
&\left(\mu^{-,0}_\la,\mu^{-,1}_\la,\mu^{-,2}_\la\right)=\left(1,e^{-\tfrac{2\pi\i}{3}},e^{\tfrac{2\pi\i}{3}}\right).
\end{aligned}
\end{equation}
The intertwining function \eqref{beta} vanishes and the normalization functions \eqref{mu} are calculated as 
$$\mu^\pm(\la)=\tfrac92,\q \la\in L_M.$$

Since the Weyl group $W$ is generated by reflections \eqref{ri}, the $\Z-$duality relation \eqref{Zdual} guarantees for any $w\in W$ that
\begin{equation*}
w\om_k\in \om_k+Q, \q k\in \{1,2\},
\end{equation*} 
and hence, for discrete values $x\in \tfrac{1}{M}(\om_j+Q)$ it holds that
$$M\langle w\om_k,x\rangle\in \langle\om_k,\om_j\rangle+\Z,\q j,k\in \{1,2\}.$$ Therefore,
the extended $C-$functions \eqref{def1} are due to \eqref{scaom} and \eqref{scash} evaluated on the decomposition \eqref{HMdis} as
\begin{equation}\label{hod1}
\Phi^\pm_\la(x)=\begin{cases}
\left(\mu^{\pm,0}_\la+\mu^{\pm,1}_\la e^{-\tfrac{2\pi\i}{3}}+\mu^{\pm,2}_\la e^{\tfrac{2\pi\i}{3}}\right)\Phi_\la(x),&x\in H_M^{(1)},\\
\left(\mu^{\pm,0}_\la+\mu^{\pm,1}_\la e^{\tfrac{2\pi\i}{3}}+\mu^{\pm,2}_\la e^{-\tfrac{2\pi\i}{3}}\right)\Phi_\la(x),&x\in H_M^{(2)}.
\end{cases}
\end{equation}
Setting the values \eqref{mu3} in formula \eqref{hod1} yields for the honeycomb $C-$functions of type III relations
\begin{equation}\label{CHIII}
\mathrm{Ch}^{+,\mathrm{III}}_\la(x)=\begin{cases}
3\Phi_\la(x)&x\in H_M^{(1)},\\
0&x\in H_M^{(2)},
\end{cases}\q 
\mathrm{Ch}^{-,\mathrm{III}}_\la(x)=\begin{cases}
0&x\in H_M^{(1)},\\
3\Phi_\la(x)&x\in H_M^{(2)}.
\end{cases}
\end{equation}
Formula \eqref{iddd} determines the values of the honeycomb $C-$functions of type III on the point set \eqref{FQM} as
\begin{equation}\label{CHIII0}
\mathrm{Ch}^{+,\mathrm{III}}_\la(x)=	\mathrm{Ch}^{-,\mathrm{III}}_\la(x)=0, \q x\in F_{Q,M}.
\end{equation}
 
Relations \eqref{CHIII} is for the honeycomb $S-$functions of similar form,
\begin{equation*}
\mathrm{Sh}^{+,\mathrm{III}}_\la(x)=\begin{cases}
3\phi_\la(x)&x\in H_M^{(1)},\\
0&x\in H_M^{(2)},
\end{cases}\q 
\mathrm{Sh}^{-,\mathrm{III}}_\la(x)=\begin{cases}
0&x\in H_M^{(1)},\\
3\phi_\la(x)&x\in H_M^{(2)},
\end{cases}
\end{equation*}
and formula \eqref{CHIII0} becomes
 \begin{equation*}
\mathrm{Sh}^{+,\mathrm{III}}_\la(x)=	\mathrm{Sh}^{-,\mathrm{III}}_\la(x)=0, \q x\in  F_{Q,M}.
\end{equation*}

The contour plots of real and imaginary parts of the honeycomb $C-$functions $\mathrm{Ch}^{+,\mathrm{III}}_\la$ are depicted in Figure \ref{ChIII}. The contour plots of real and imaginary parts of the honeycomb $S-$functions $\mathrm{Sh}^{+,\mathrm{III}}_\la$ are depicted in Figure \ref{ShIII}. 
\begin{figure}
\includegraphics{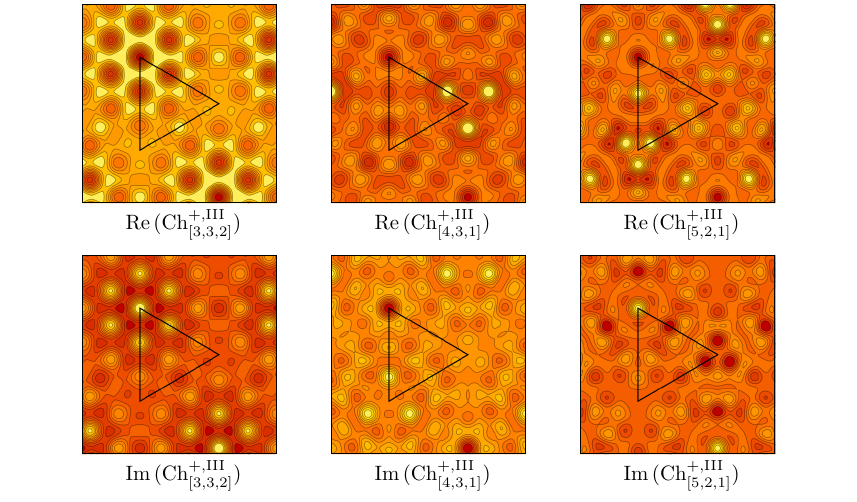}
\caption{{\small The contour plots of the honeycomb $C-$functions $\mathrm{Ch}^{+,\mathrm{III}}_\la$ with $M=8$. The triangle depicts the fundamental domain $F_Q$.}}
\label{ChIII}
\end{figure}

\begin{figure}
\includegraphics{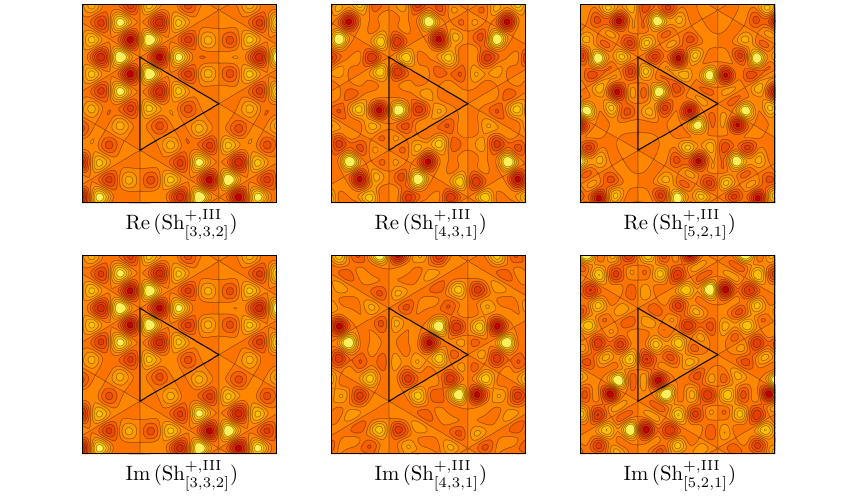}
\caption{{\small The contour plots of the honeycomb $S-$functions $\mathrm{Sh}^{+,\mathrm{III}}_\la$ with $M=8$. The triangle depicts the fundamental domain $F_Q$.}}
\label{ShIII}
\end{figure}

\section{Discrete honeycomb lattice transforms}

\subsection{Four types of discrete transforms}\

Two interpolating functions $\mathrm{I}[f]_M:\R^2\map \C$ and $\mathrm{Ih}[f]_M:\R^2\map \C$ of any function sampled on the honeycomb lattice fragment $f\in \mathcal{H}_M$ are defined as linear combinations of the honeycomb and Hartley honeycomb $C-$functions,
\begin{align}
\mathrm{I}[f]_M(x)=& \sum_{\la\in L_M} \left(c^+_\la \mathrm{Ch}^+_\la(x) + c^-_\la \mathrm{Ch}_\la^-(x)\right),\label{intc} \\
\mathrm{Ih}[f]_M(x)=& \sum_{\la\in L_M} \left(d^+_\la \mathrm{Cah}^+_\la(x) + d^-_\la \mathrm{Cah}_\la^-(x)\right), \label{hintc}
\end{align}
that coincide with the function $f$ on the interpolation nodes,
\begin{align*}
\mathrm{I}[f]_M(s)=& f(s), \q s\in H_M, \\
\mathrm{Ih}[f]_M(s)=& f(s), \q s\in H_M. 
\end{align*}
The frequency spectrum coefficients $c^\pm_\la$ and $d^\pm_\la$ are uniquely determined by Theorems \ref{thmdis} and \ref{thmdish} and calculated as the standard Fourier coefficients,
\begin{align}
c^\pm_\la=& \frac{\sca{f}{\mathrm{Ch}^\pm_\la}_{H_M}}{\sca{\mathrm{Ch}^\pm_\la}{\mathrm{Ch}^\pm_\la}_{H_M}}=(12M^2h_M(\la)\mu^\pm(\la))^{-1}\sum_{s\in H_M}\ep(s) f(s)\overline{\mathrm{Ch}^\pm_\la(s)},\label{ctrans}\\
d^\pm_\la=& \frac{\sca{f}{\mathrm{Cah}^\pm_\la}_{H_M}}{\sca{\mathrm{Cah}^\pm_\la}{\mathrm{Cah}^\pm_\la}_{H_M}}=(12M^2h_M(\la)\mu^\pm(\la))^{-1}\sum_{s\in H_M}\ep(s) f(s)\overline{\mathrm{Cah}^\pm_\la(s)}.\label{hctrans}
\end{align}
The corresponding Plancherel formulas are also valid,
\begin{align*}
\sum_{s\in H_M}\ep(s)\abs{f(s)}^2&=12M^2\sum_{\la\in L_M}h_M(\la)\left(\mu^+(\la)\abs{c^+_\lambda}^2+\mu^-(\la)\abs{c^-_\lambda}^2\right),\\
\sum_{s\in H_M}\ep(s)\abs{f(s)}^2&=12M^2\sum_{\la\in L_M}h_M(\la)\left(\mu^+(\la)\abs{d^+_\lambda}^2+\mu^-(\la)\abs{d^-_\lambda}^2\right).
\end{align*}
Formulas \eqref{ctrans}, \eqref{hctrans} and \eqref{intc}, \eqref{hintc} provide forward and backward Fourier-Weyl and Hartley-Weyl honeycomb $C-$transforms, respectively.

Two interpolating functions $\mathrm{\wt{I}}[f]_M:\R^2\map \C$ and $\mathrm{\wt{I}h}[f]_M:\R^2\map \C$ of any function sampled on the interior honeycomb lattice fragment $\wt{\mathcal{H}}_M$ are defined as linear combinations of the honeycomb and Hartley honeycomb $S-$functions,
\begin{align}
\mathrm{\wt{I}}[f]_M(x)=& \sum_{\la\in \wt{L}_M} \left(\wt{c}^+_\la \mathrm{Sh}^+_\la(x) + \wt{c}^-_\la \mathrm{Sh}_\la^-(x)\right), \label{ints} \\
\mathrm{\wt{I}h}[f]_M(x)=& \sum_{\la\in \wt{L}_M} \left(\wt{d}^+_\la \mathrm{Sah}^+_\la(x) + \wt{d}^-_\la \mathrm{Sah}_\la^-(x)\right), \label{hints} 
\end{align}
which coincide with the function $f$ on the interpolation nodes,
\begin{align*}
\mathrm{\mathrm{\wt{I}}}[f]_M(s)=& f(s), \q s\in \wt{H}_M, \\
\mathrm{\wt{I}h}[f]_M(s)=& f(s), \q s\in \wt{H}_M. 
\end{align*}
The frequency spectrum coefficients $\wt c^\pm_\la$ and $\wt d^\pm_\la$ are uniquely determined by Theorem \ref{thmdis2} and calculated as the standard Fourier coefficients,
\begin{align}
\wt{c}^\pm_\la=& \frac{\sca{f}{\mathrm{Sh}^\pm_\la}_{\wt{H}_M}}{\sca{\mathrm{Sh}^\pm_\la}{\mathrm{Sh}^\pm_\la}_{\wt{H}_M}}=(2M^2\mu^\pm(\la))^{-1}\sum_{s\in \wt{H}_M} f(s)\overline{\mathrm{Sh}^\pm_\la(s)},\label{strans}\\
\wt{d}^\pm_\la=& \frac{\sca{f}{\mathrm{Sah}^\pm_\la}_{\wt{H}_M}}{\sca{\mathrm{Sah}^\pm_\la}{\mathrm{Sah}^\pm_\la}_{\wt{H}_M}}=(2M^2\mu^\pm(\la))^{-1}\sum_{s\in \wt{H}_M} f(s)\overline{\mathrm{Sah}^\pm_\la(s)}.\label{hstrans}
\end{align}
The corresponding Plancherel formulas are also valid,
\begin{align*}
\sum_{s\in \wt{H}_M}\abs{f(s)}^2&=2M^2\sum_{\la\in \wt{L}_M}\left(\mu^+(\la)\abs{\wt{c}^+_\lambda}^2+\mu^-(\la)\abs{\wt{c}^-_\lambda}^2\right),\\
\sum_{s\in \wt{H}_M}\abs{f(s)}^2&=2M^2\sum_{\la\in \wt{L}_M}\left(\mu^+(\la)\abs{\wt{d}^+_\lambda}^2+\mu^-(\la)\abs{\wt{d}^-_\lambda}^2\right).
\end{align*}
Formulas \eqref{strans}, \eqref{hstrans} and \eqref{ints}, \eqref{hints} provide forward and backward Fourier-Weyl and Hartley-Weyl honeycomb $S-$transforms, respectively.

\begin{example} [Interpolation tests]\label{exint}

As a specific model function, the following real-valued function is defined on the fundamental domain $F_Q$ for any point $x=x_1\om_1+x_2\om_2$ in $\om-$basis, 
$$f(x)=0.4\,e^{-\tfrac{1}{4\sigma^2}\left(\left(x_1 - \tfrac{1}{3}\right)^2 + \tfrac13\left(x_1+2x_2 - 1\right)^2\right)}.$$
The 2D graph of the model function $f$, with $\sigma=0.065$ fixed, is plotted in Figure \ref{exp}.
\begin{figure}
\includegraphics{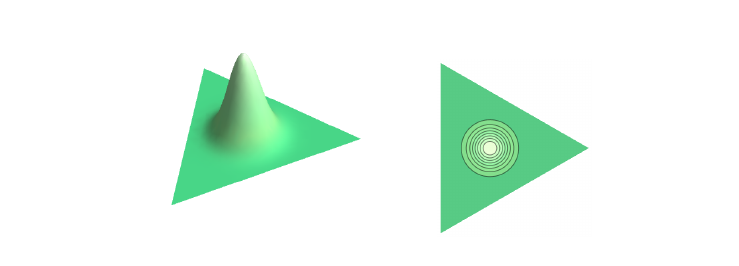}
\caption{The model function $f$ is plotted over the fundamental domain $F_Q$.}
\label{exp}
\end{figure}

The function $f$ is interpolated by the honeycomb Hartley $C-$ and $S-$functions \eqref{hintc}, \eqref{hints} of types I and II.
The interpolating functions $\mathrm{Ih}^\mathrm{I}[f]_M$ and $\mathrm{Ih}^\mathrm{II}[f]_M$, corresponding to the honeycomb Hartley $C-$functions \eqref{typeI} and \eqref{typeII}, are plotted in Figures \ref{approx1} and \ref{approx2}. The interpolating functions $\mathrm{\wt{I}h^I}[f]_M$ and $\mathrm{\wt{I}h^{II}}[f]_M$, corresponding to the honeycomb Hartley $S-$functions \eqref{typeI} and \eqref{typeII}, are plotted in Figures \ref{approx3} and \ref{approx4}. Integral error estimates of all four types of interpolations are calculated in Table \ref{errors}. 
\begin{figure}
\includegraphics{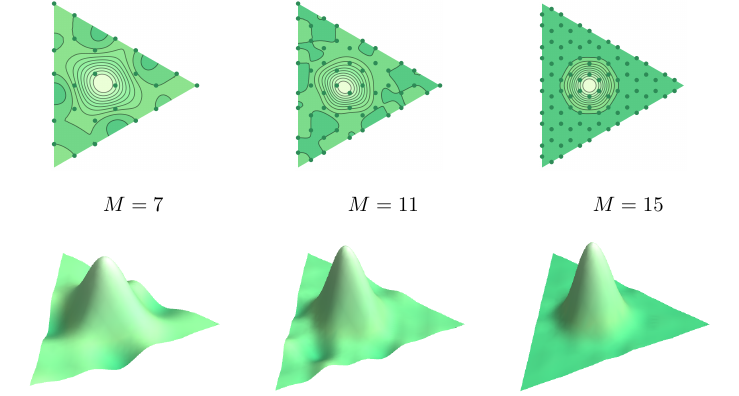}
\caption{The interpolations $\mathrm{Ih}^\mathrm{I}[f]_M$ are for $M=7,11,15$ plotted over the fundamental domain $F_Q$. Sampling points for the interpolations are depicted as green dots.}
\label{approx1}
\end{figure}

\begin{figure}
\includegraphics{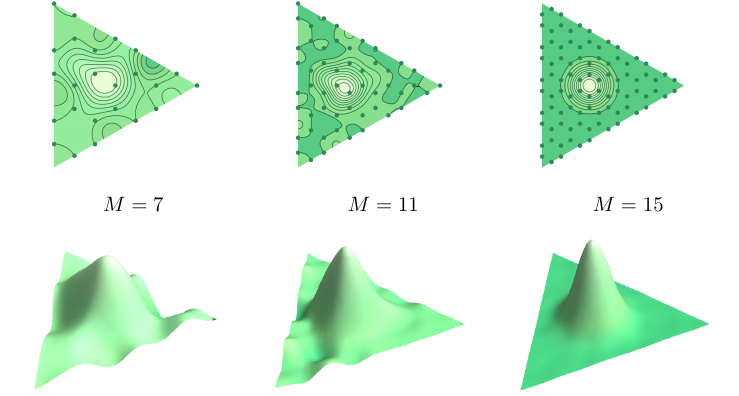}
\caption{The interpolations $\mathrm{Ih}^\mathrm{II}[f]_M$ are for $M=7,11,15$ plotted over the fundamental domain $F_Q$. Sampling points for the interpolations are depicted as green dots.}
\label{approx2}
\end{figure}

\begin{figure}
\includegraphics{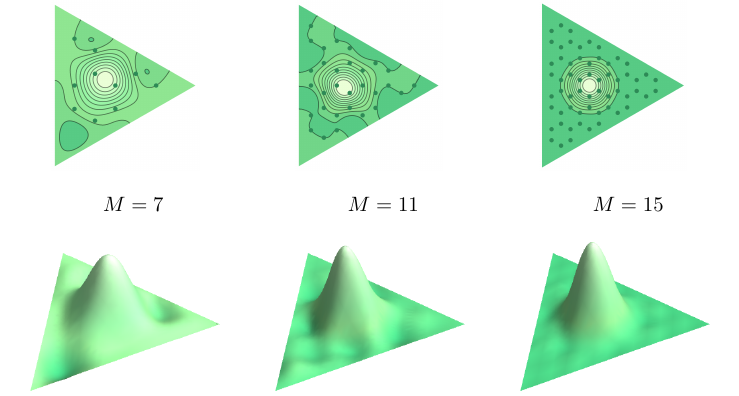}
\caption{The interpolations $\mathrm{\wt{I}h^I}[f]_M$ are for $M=7,11,15$ plotted over the fundamental domain $F_Q$. Sampling points for the interpolations are depicted as green dots.}
\label{approx3}
\end{figure}

\begin{figure}
\includegraphics{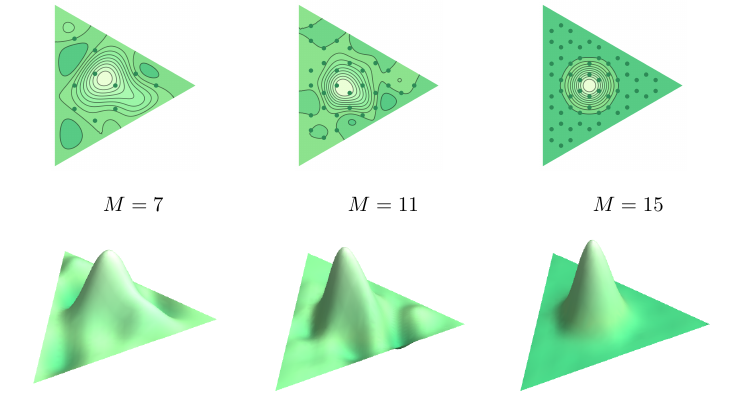}
\caption{The interpolations $\mathrm{\wt{I}h^{II}}[f]_M$ are for $M=7,11,15$ plotted over the fundamental domain $F_Q$. Sampling points for the interpolations are depicted as green dots.}
\label{approx4}
\end{figure}

\bgroup
\def\arraystretch{1.3}
\begin{table}[ht]
\begin{tabular}{|c||c|c|c|c|c|}
\hline
$M$&$7$&$9$&$11$&$13$&$15$\\ \hline\hline
$\int_{F_Q}|f-\mathrm{Ih}^\mathrm{I}[f]_M|^2$&$2108\times10^{-7}$&$2663\times 10^{-7}$&$454\times10^{-7}$&$52\times 10^{-7}$&$32\times10^{-7}$\\\hline
$\int_{F_Q}|f-\mathrm{Ih}^\mathrm{II}[f]_M|^2$&$4964\times 10^{-7}$&$2208\times 10^{-7}$&$950\times 10^{-7}$&$106\times 10^{-7}$&$8\times 10^{-7}$\\\hline
$\int_{F_Q}|f-\mathrm{\wt{I}h^I}[f]_M|^2$&$3177\times10^{-7}$&$2794\times 10^{-7}$&$462\times10^{-7}$&$53\times 10^{-7}$&$32\times10^{-7}$\\\hline
$\int_{F_Q}|f-\mathrm{\wt{I}h^{II}}[f]_M|^2$&$5666\times 10^{-7}$&$1166\times10^{-7}$&$1054\times 10^{-7}$&$112\times10^{-7}$&$11\times 10^{-7}$\\\hline
\end{tabular}
\bigskip
\caption{{\small The integral error estimates of the interpolations $\mathrm{Ih}^\mathrm{I}[f]_M$, $\mathrm{Ih}^\mathrm{II}[f]_M$, $\mathrm{\wt{I}h^I}[f]_M$ and $\mathrm{\wt{I}h^{II}}[f]_M$ are tabulated for $M=7,9,11,13,15$.}}
\label{errors}
\end{table}
\egroup
\end{example}

\subsection{Matrices of normalized discrete transforms}\

The points and weights in the sets $H_M$, $\wt H_M$ and $L_M$, $\wt L_M$ are ordered according to the lexicographic order of coordinates \eqref{kac} and \eqref{kacla}. 
For the honeycomb $C-$transforms, four matrices $\mathbb{I}^{M,\pm}$,  $\mathbb{I}\mathrm{h}^{M,\pm}$ are defined by relations
\begin{align}
\mathbb{I}^{M,\pm}_{\la s}=& \sqrt{\epsilon(s)\left(12M^2h_M(\la)\mu^\pm(\la)\right)^{-1}}\,\overline{\mathrm{Ch}^\pm_\la(s)}, \q \la \in L_M, s\in H_M,\nonumber\\
\mathbb{I}\mathrm{h}^{M,\pm}_{\la s}=&\sqrt{\epsilon(s)\left(12M^2h_M(\la)\mu^\pm(\la)\right)^{-1}}\,\overline{\mathrm{Cah}^\pm_\la(s)}, \q \la \in L_M, s\in H_M. \label{Ihls}
\end{align}
The unitary transform matrices $\mathbb{I}_M$ and $\mathbb{I}\mathrm{h}_M$, assigned to the normalized Fourier-Weyl and Hartley-Weyl honeycomb $C-$transforms, are given as the following block matrices,
\begin{align}\label{IhM}
\mathbb{I}_M=&\begin{pmatrix} \mathbb{I}^{M,+}\\ \mathbb{I}^{M,-}\end{pmatrix}, \q \mathbb{I}\mathrm{h}_M=\begin{pmatrix} \mathbb{I}\mathrm{h}^{M,+}\\ \mathbb{I}\mathrm{h}^{M,-}\end{pmatrix}.
\end{align}

For the honeycomb $S-$transforms, four matrices $\wt{\mathbb{I}}^{M,\pm}$, $\wt{\mathbb{I}}\mathrm{h}^{M,\pm}$ are defined by relations
\begin{align}
\wt{\mathbb{I}}^{M,\pm}_{\la s} &=\sqrt{\epsilon(s)\left(2M^2h_M(\la)\mu^\pm(\la)\right)^{-1}}\,\overline{\mathrm{Sh}^\pm_\la(s)}, \q \la \in \wt L_M, s\in \wt H_M,\nonumber \\
\wt{\mathbb{I}}\mathrm{h}^{M,\pm}_{\la s} &=\sqrt{\epsilon(s)\left(2M^2h_M(\la)\mu^\pm(\la)\right)^{-1}}\,\overline{\mathrm{Sah}^\pm_\la(s)}, \q \la \in \wt L_M, s\in \wt H_M. \label{Ihlss}
\end{align}
The unitary transform matrices $\wt{\mathbb{I}}_M$ and $\wt{\mathbb{I}}\mathrm{h}_M$, assigned to the normalized Fourier-Weyl and Hartley-Weyl honeycomb $S-$transforms, are given as the following block matrices,
\begin{align}\label{IhsM}
\wt{\mathbb{I}}_M=&\begin{pmatrix} \wt{\mathbb{I}}^{M,+}\\ \wt{\mathbb{I}}^{M,-}\end{pmatrix}, \q \wt{\mathbb{I}}\mathrm{h}_M=\begin{pmatrix} \wt{\mathbb{I}}\mathrm{h}^{M,+}\\ \wt{\mathbb{I}}\mathrm{h}^{M,-}\end{pmatrix}.
\end{align}

\begin{example}[Transform matrices $\mathbb{I}\mathrm{h}_4$ and $\wt{\mathbb{I}}\mathrm{h}_7$]
The lexicographically ordered point set $H_4$ is of the form 
$$H_4 = \left\{\left[0, 0, 1\right], \left[0, \tfrac{1}{4}, \tfrac{3}{4}\right], \left[0, \tfrac{3}{4}, \tfrac{1}{4}\right], \left[0, 1, 0\right], \left[\tfrac{1}{4}, \tfrac{1}{4}, \tfrac{1}{2}\right], \left[\tfrac{1}{4}, \tfrac{1}{2}, \tfrac{1}{4}\right], \left[\tfrac{1}{2}, 0, \tfrac{1}{2}\right], \left[\tfrac{1}{2}, \tfrac{1}{2}, 0\right], \left[\tfrac{3}{4}, 0, \tfrac{1}{4}\right], \left[\tfrac{3}{4}, \tfrac{1}{4}, 0\right]\right\},$$
and the lexicographically ordered weight set $L_4$ contains the following weights,
$$L_4=\{[2, 1, 1], [2, 2, 0], [3, 0, 1], [3, 1, 0], [4, 0, 0]\}.$$
The unitary transform matrix $\mathbb{I}\mathrm{h}^\mathrm{II}_4$, corresponding to the honeycomb Hartley $C-$functions~\eqref{typeII}, is computed from relations \eqref{Ihls} and \eqref{IhM} as  
{\small
$$
\mathbb{I}\mathrm{h}^\mathrm{II}_4=\begin{pmatrix}
0.433& 0.250& 0.250& 0.433& -0.354& -0.354& -0.250& -0.250& 0.250& 0.250\\
 -0.306& -0.177& 0.177& 0.306& 0.250& -0.250&0.530& -0.530& 0.177& -0.177\\
-0.421& -0.544& 0.128& 0.099& -0.344&  0.081& -0.057& 0.243& 0.358& 0.429\\
0.099&  0.128& -0.544& -0.421& 0.081& -0.344& 0.243& -0.057& 0.429&  0.358\\
0.177& 0.306& 0.306& 0.177& 0.433& 0.433& 0.306& 0.306& 
  0.306& 0.306\\
-0.433& 0.250& -0.250& 0.433& 0.354&-0.354& -0.250& 0.250& -0.250& 0.250\\
 -0.306& 0.177& 0.177& -0.306& 0.250&0.250& -0.530& -0.530& 0.177&0.177\\
 -0.099& 0.128& 0.544& -0.421& -0.081& -0.344& 0.243&0.057& -0.429& 0.358\\
0.421& -0.544& -0.128&  0.099& 0.344& 0.081& -0.057& -0.243& -0.358& 0.429\\
-0.176& 0.306& -0.306& 0.176& -0.433& 0.433& 0.306& -0.306& -0.306& 0.306
\end{pmatrix}.
$$}

The lexicographically ordered  interior point set $\wt H_7$ is of the form 
$$\wt{H}_7= \left\{\left[\tfrac{1}{7}, \tfrac{1}{7}, \tfrac{5}{7}\right], \left[\tfrac{1}{7}, \tfrac{2}{7}, \tfrac{4}{7}\right], \left[\tfrac{1}{7}, \tfrac{4}{7}, \tfrac{2}{7}\right], \left[\tfrac{1}{7}, \tfrac{5}{7}, 
   \tfrac{1}{7}\right], \left[\tfrac{2}{7},\tfrac{ 2}{7}, \tfrac{3}{7}\right], \left[\tfrac{2}{7}, \tfrac{3}{7}, \tfrac{2}{7}\right], \left[\tfrac{3}{7}, \tfrac{1}{7}, \tfrac{3}{7}\right], \left[\tfrac{3}{7}, \tfrac{3}{7},
    \tfrac{1}{7}\right], \left[\tfrac{4}{7}, \tfrac{1}{7}, \tfrac{2}{7}\right], \left[\tfrac{4}{7}, \tfrac{2}{7}, \tfrac{1}{7}\right]\right\},$$
and the lexicographically ordered weight set $\wt L_7$ contains the following weights,	
$$\wt{L}_7 = \{[3, 2, 2], [3, 3, 1], [4, 1, 2], [4, 2, 1], [5, 1, 1]\}.$$   
The unitary transform matrix  $\wt{\mathbb{I}}\mathrm{h}^\mathrm{II}_7$, corresponding to the honeycomb Hartley $S-$functions~\eqref{typeII}, is computed from relations \eqref{Ihlss} and \eqref{IhsM} as       
{\small 
$$
\wt{\mathbb{I}}\mathrm{h}^\mathrm{II}_7=\begin{pmatrix}
-0.482& -0.267& -0.267& -0.482& 0.333& 0.333& 0.119&0.119& -0.267& -0.267\\
-0.333& -0.267& 0.267& 0.333& 0.119& -0.119& 0.482& -0.482& 0.267& -0.267\\
-0.068& -0.096& 0.526& 0.372& -0.068& 0.372& -0.372&0.068& -0.458& -0.276\\
0.372& 0.526& -0.096& -0.068& 0.372& -0.068&0.068& -0.372& -0.276& -0.458\\
-0.119& -0.267& -0.267& -0.119& -0.482& -0.482& -0.333& -0.333& -0.267& -0.267\\
-0.482& 0.267& -0.267& 0.482& 0.333& -0.333& -0.119&  0.119& -0.267& 0.267\\
 0.333& -0.267& -0.267& 0.333& -0.119& -0.119&  0.482&0.482& -0.267& -0.267\\
 0.372& -0.526& -0.096& 0.068& 0.372&  0.068& -0.068& -0.372& -0.276& 0.456\\
 -0.068& 0.096& 0.526&-0.372& -0.068& -0.372& 0.372& 0.068& -0.458&0.276\\
 -0.119& 0.267& -0.267&  0.119& -0.482& 0.482& 0.333& -0.333& -0.267& 0.267
\end{pmatrix}.
$$}

\end{example}

\section{Concluding Remarks}

\begin{itemize}
\item Excellent interpolating behaviour of the Hartley honeycomb $C-$ and $S-$functions of types I and II in Example \ref{exint} promises similar success of related
digital data processing techniques. As is depicted in Table \ref{errors} of integral error estimates, the more complicated type II honeycomb functions slightly outperform the simpler functions of type I. Since  the unitary transform matrices \eqref{IhM} and \eqref{IhsM} are for any fixed $M\in \N$ directly precalculated, the complexity of the given type of function is of minor consequence. Unlike types I and II, the honeycomb functions of type III vanish on the point sets \eqref{FQM} and are suitable for interpolation of functions with the same property. Formulation of general convergence criteria, depending necessarily on the extension coefficients $\mu^{\pm,k}_\la$ of the honeycomb functions, poses an open problem.   
\item The notation used for the honeycomb orbit functions is motivated by transversal vibrational modes of the mechanical graphene model \cite{CsTi,DrSa}. In Figure \ref{honey}, let the lines linking the dots represent the springs of spring constants $\kappa$ and natural lenghts $l_0$. The dots depict the points of masses $m$ with the equilibrium distance between the two nearest points denoted by $R_0$.  The parameter $\eta=l_0/R_0$, $\eta<1$ determines stretching of the system. 
Thus, the honeycomb $C-$ and $S-$functions~\eqref{typeII} of type II represent transversal eigenvibrations of this model subjected to discretized von Neumann and Dirichlet boundary conditions on the depicted triangle, respectively. The frequencies corresponding to the modes $\mathrm{Cah}^{\pm,\mathrm{II}}_{\la}$, $\la \in L_M$ and $\mathrm{Sah}^{\pm,\mathrm{II}}_{\la}$, $\la \in \wt L_M$ are given as
\begin{equation}\label{frek}
\omega_\la^{\pm}=\sqrt{\frac{\kappa(1-\eta)}{m}\left(3\pm\frac{1}{2}\abs{\Phi_\la\left(\frac{\omega_1}{M}\right)}\right)}.	
\end{equation}
However, full exposition of the method for calculating the frequencies $\omega_\la^{\pm}$ and extension coefficients~\eqref{typeII} requires a separated article. The eigenfrequencies \eqref{frek} correspond to the frequency spectrum in \cite{Roz}. The construction of the discrete eigenfunctions in \cite{Roz} leads to two separate descriptions of the values of the modes at lattice points \eqref{HM1} and \eqref{HM2}. On the other hand, the presented subtractive approach yields uniform description of each mode by one discretized function $\mathrm{Sah}^{\pm,\mathrm{II}}_{\la}$, $\la \in \wt L_M$.
\item The set of the honeycomb $C-$ and $S-$functions, depending on the six parameters $\mu^{\pm,k}_\la$ for each $\la\in L_M$ and $\la\in \wt L_M$, comprises solutions of the three non-linear conditions \eqref{cond1} and \eqref{cond2}. Other cases of type I honeycomb functions constitute for instance
\begin{align*}
&\left(\mu^{+,0}_\la,\mu^{+,1}_\la,\mu^{+,2}_\la\right)=(0,1,0),\\
&\left(\mu^{-,0}_\la,\mu^{-,1}_\la,\mu^{-,2}_\la\right)=(1,0,-1),
\end{align*}
as well as
\begin{align*}
&\left(\mu^{+,0}_\la,\mu^{+,1}_\la,\mu^{+,2}_\la\right)=(0,0,1),\\
&\left(\mu^{-,0}_\la,\mu^{-,1}_\la,\mu^{-,2}_\la\right)=(1,-1,0).
\end{align*}
Finding a suitable equivalence relation on the set of solutions and describing the entire set of the honeycomb orbit functions up to this equivalence represents an unsolved problem. Generalization of the presented subtractive method for construction of the discretely orthogonal parametric systems of extended Weyl orbit functions to the triangular honeycomb dot with zigzag boundaries and to other crystallographic root systems also pose open problems.   
\item The families of $C-$ and $S-$functions induce two kinds of discretely orthogonal generalized Chebyshev polynomials. Cubature formulas for numerical integration are among the recently studied associated polynomial methods \cite{HMPcub,MMP,MPcub}. As a linear combination of $C-$ and $S-$ functions, each case of the honeycomb functions generates a set of polynomials discretely orthogonal on points forming a deformed honeycomb pattern inside the Steiner's hypocycloid. Properties of these polynomials and the related polynomial methods deserve further study. Existence of variants of Macdonald polynomials \cite{diejen} orthogonal on  the deformed honeycomb pattern poses another open problem. The functions symmetric with respect to even subgroups of Weyl groups generalize the standard Weyl and Hartley-Weyl orbit functions \cite{KP3,HJedis}. The reflection group $A_2$ admits one further type of these $E-$functions. The explicit form of their root-lattice discretization and the honeycomb modification deserve further study.     
\end{itemize}

\section*{Acknowledgments}
This work was supported by the Grant Agency of the Czech Technical University in Prague, grant number SGS16/239/OHK4/3T/14.  LM and JH gratefully acknowledge the support of this work by RVO14000.

\end{document}